\documentclass[aps,english,10pt,a4,superscriptaddress]{revtex4}  
\newcommand{\comment}[1]{}

\usepackage{amsmath}
\usepackage{amsthm}
\usepackage{amssymb}
\usepackage{amsfonts}
\usepackage{bbm}
\usepackage{epsfig}
\usepackage{times}
\usepackage{babel}
\usepackage{color}
\usepackage{hyperref}
\usepackage{framed}

\newcommand{\mean}{\mathbb{E}}

\newcommand{\F}{\mathbb{F}}
\newcommand{\id}{\mathbf{1}}
\newcommand{\R}{\mathbb{R}}
\newcommand{\C}{\mathbb{C}}
\newcommand{\Id}{\mathbb{I}}
\newcommand{\G}{\mathcal{G}}
\newcommand{\p}{\mathcal{P}}
\newcommand{\N}{\mathbb{N}}

\newcommand{\Tr}{\operatorname{Tr}}

\graphicspath{{../_fig/}}

\newtheorem{theorem}{Theorem}
\newtheorem{definition}[theorem]{Definition}
\newtheorem{lemma}[theorem]{Lemma}

\newtheorem{example}[theorem]{Example}
\newtheorem{corollary}[theorem]{Corollary}

\begin{document}

\title{Unifying typical entanglement and coin tossing: on randomization in probabilistic theories}

\author{Markus P.\ M\"uller}
\affiliation{Perimeter Institute for Theoretical Physics, 31 Caroline Street North, Waterloo, ON N2L 2Y5, Canada}
\author{Oscar C.\ O.\ Dahlsten}
\affiliation{Centre for Quantum Technologies, National University of Singapore, 3 Science Drive 2, Singapore 117543, Singapore}
\affiliation{Clarendon Laboratory, University of Oxford, Parks Road, Oxford OX1 3PU, United Kingdom}
\author{Vlatko Vedral}
\affiliation{Centre for Quantum Technologies, National University of Singapore, 3 Science Drive 2, Singapore 117543, Singapore}
\affiliation{Clarendon Laboratory, University of Oxford, Parks Road, Oxford OX1 3PU, United Kingdom}
\affiliation{Department of Physics, National University of Singapore, 2 Science Drive 3, Singapore 117542, Singapore}

\date{\today}

\begin{abstract}
It is well-known that pure quantum states are typically almost maximally entangled, and thus have close to maximally mixed subsystems. We consider whether this is true for probabilistic theories more generally, and not just for quantum theory. 
We derive a formula for the expected purity of a subsystem in any probabilistic theory for which this quantity is well-defined.
It applies to typical entanglement in pure quantum states, coin tossing in classical probability theory, and randomization in post-quantum theories; a simple generalization yields the typical entanglement in (anti)symmetric quantum subspaces.
The formula is exact and simple, only containing the number of degrees of freedom and the information capacity of the respective systems. It allows us to generalize statistical  physics arguments in a way which depends only on coarse properties of the underlying theory. The proof of the formula generalizes several randomization notions to general probabilistic theories. This includes a generalization of purity, contributing to the recent effort of finding appropriate generalized entropy measures.
\end{abstract}

\maketitle

\tableofcontents

\parskip .75ex

\section{Introduction}
It is increasingly recognized that entanglement is ubiquitous, as opposed to a rare resource that is difficult to create. In fact most unitary time evolutions (in a sense to be made precise later) generate a large amount of entanglement within a closed quantum system. This turns out to be equivalent to saying that {\em pure quantum states are typically almost maximally entangled}.

This striking observation was already made decades ago, see e.g.~\cite{Lubkin78, LloydP88, Page93, FoongK94}, although it was initially phrased as `subsystem entropy typically being maximal'---this was before subsystem entropy became the canonical measure of entanglement for pure states. The observation and its subsequent refinements have helped us understand more about entanglement and its role in information processing~\cite{FoongK94,Page93,LloydP88,Lubkin78,HaydenLW06,HarrowHL04,AbeyesingheHS06,SerafiniDP06,SerafiniDP07,SmithL05,DahlstenP06,DahlstenPhD,OliveiraDP07,Dupuis09} as well as statistical mechanics~\cite{HaydenP07,CalsamigliaHDB07,GemmerOM01,PopescuSW06,Lloyd88,GemmerMM04,LubkinL93,MuellerGE10,delRioARDV11}. For example, bearing the above in mind it is not surprising that the difficulty for an experimenter trying to perform e.g.\ quantum computing is not to {\em generate} entanglement but to {\em control} what is entangled with what, and in particular to avoid entanglement between the experiment and the environment, as that will increase the entropy of the system.

Here we show that this observation is an instance of a more universal phenomenon which appears in a wide class of probabilistic theories: systems typically randomize locally if a global transformation is applied.
More specifically, the expected amount of randomization can be expressed by a simple formula, which is universally valid for any probabilistic theory satisfying a small set of requirements. The formula describes classical coin tossing as well as typical entanglement in quantum and possible post-quantum theories, and has a particularly simple form which does not depend on the details of the theory.

We work in the framework of generalized probabilistic theories, also known as the ``convex framework''. This amounts to taking an operational pragmatic point of view that the physical content of a theory is the predictions of outcome statistics, conditional on the experimental settings. A wide range of theories can be formulated in this framework, including quantum theory and classical probability theory.

We ask how pure or mixed subsystems tend to be in such theories, if the global state is drawn randomly (possibly subject to some constraints). To make the question well-defined, we add some additional restrictions on the set of theories, including crucially that all pure states are connected by reversible dynamics.  Our main result is to give a simple expression for the expected value of the purity of a subsystem in such probabilistic theories.
The expression shows that, in certain limits, subsystems are typically close to maximally random. (In the case of pure global quantum states this is equivalent to saying that the states are typically close to maximally entangled).

Our result unifies several instances of randomization associated with different theories. It also clarifies which features of the theory are behind this phenomenon and govern the strength with which it occurs. Some of the techniques invented in the proof are in addition interesting in themselves. These include generalizations of the notions of purity and of Pauli operators to general probabilistic theories. The proof is moreover guided by an intuitive Heisenberg-picture argument which is different to the standard arguments for the quantum case and arguably adds to our understanding of the quantum result.

We apply the result to generalize a specific statistical mechanical argument employing typical entanglement which is related to the second law of thermodynamics. We moreover calculate the typical subsystem purity in a variety of cases, including typical entanglement of pure symmetric and antisymmetric bipartite quantum states, which is to our knowledge also a new contribution. 

The presentation is divided into two parts.  It should be possible for readers not wishing to familiarize themselves with general probabilistic theories to only read the first part. The first part describes the main results and their implications with an emphasis on the quantum and classical cases.  In the second part we deal with the general probabilistic case.

\section{Main results and overview}
One of our main results is an identity which relates simple properties of state spaces to the randomization of subsystems.
Suppose that Alice and Bob hold a bipartite system $AB$ (for example, a composite quantum system $A\otimes B$).
They draw a biparte state $\omega^{AB}$ at random; it may be a random pure state, or a random mixed state with fixed purity $\p(\omega^{AB})$.
Then, the reduced state at Alice, $\omega^A$, will in general be mixed: its
expected purity turns out to be
\begin{equation}
   \mathbb{E}_\omega\p(\omega^A)=\frac{K_A-1}{K_A K_B-1}\cdot\frac{N_A N_B-1}{N_A-1}\cdot\p(\omega^{AB}).
   \label{eqMain}
\end{equation}
The parameters $K_A$ and $N_A$ denote the state space dimension and information carrying capacity of $A$, respectively (similarly for $B$).
It will turn out that this simple formula describes the \emph{typical amount of entanglement in random pure quantum states} (in
particular the fact that most quantum states are almost maximally entangled), and at the same time \emph{classical
coin tossing}.

Moreover, this identity describes randomization in possible probabilistic theories beyond quantum theory. It shows that
very coarse properties of a theory are sufficient to determine its randomization power -- basically, the ratio between the
total number of degrees of freedom $K$ versus the number of perfectly distinguishable states $N$. A generalization of this identity
gives the expected amount of entanglement in symmetric and antisymmetric subspaces, a quantum result that seems to
be new as well.

In this section, we give a self-contained and elementary statement of our results:
\begin{itemize}
\item First, we outline how we define the purity $\p$ in general (Subsection~\ref{SubsecPurity}), and we explain the notions of ``state space dimension $K$`` and ``capacity $N$'' (in Subsection~\ref{SubsecKN}).
\item Then we demonstrate how our result unifies typical quantum entanglement and coin tossing
(Subsection~\ref{SubsecUnify}) into a single identity, and we apply a simple generalization of this result to compute the average entanglement in symmetric and
antisymmetric quantum subspaces (Subsection~\ref{SubsecSymm}).
\item In Subsection~\ref{SubsecStatPhys}, we apply our results to statistical physics. We argue that the results contribute to a
theory-independent understanding of some aspects of thermalization and the second law, which may be applied in situations like
black hole thermodynamics where the underlying probabilistic theory is not fully known.
\item Finally, we give a simple proof of the quantum case in Subsection~\ref{SubsecQProof}, which also illustrates the main ideas of
the more general proof in Section~\ref{SecMath}.
\end{itemize}
The detailed mathematical calculations and results are given in Section~\ref{SecMath}. The main result is Theorem~\ref{TheMainMain},
which contains the exact list of assumptions which must be satisfied for eq.~(\ref{eqMain}) to hold. There is also a more general version of this
result which needs less assumptions, but is slightly less intuitive (Theorem~\ref{TheMain1}). An even more general version concerns random states
under constraints (Theorem~\ref{TheFace}); this one can be used to derive the average entanglement in (anti)symmetric subspaces in quantum theory.

Section~\ref{SecMath} uses the mathematical framework of general probabilistic theories, as explained for example in~\cite{Barrett07,Barnum}.
Several results in this section are of independent interest in this framework. In particular, we introduce and analyze a general-probabilistic
notion of \emph{purity}. Due to its group-theoretic origin, purity satisfies several interesting identities.
It can be seen as an easy-to-compute replacement for entropy, and has several advantages over recently proposed entropy measures for
probabilistic theories (cf.\ Subsection~\ref{SubsecComparison}).

The remainder of this section does \emph{not} assume familiarity with general probabilistic theories.

\subsection{Purity}
\label{SubsecPurity}
In quantum theory, the standard notion of purity of a density matrix $\rho$ with eigenvalues $\{\lambda_i  \}$ is $\Tr(\rho^2)=\sum_i \lambda_i^2$.
This quantity has an operational meaning as the probability that two successive measurements on two identical copies of $\rho$ give
the same outcome, if one measures in the basis where $\rho$ is diagonal, i.e.\ in the minimal uncertainty basis. It is therefore sometimes
called the {\em collision probability}.

In this work, it will turn out to be extremely useful to rescale this quantity slightly. For density matrices $\rho$ on $\C^n$, we define
\begin{equation}
   \p(\rho)=\frac n {n-1}\Tr(\rho^2)-\frac 1 {n-1}.
   \label{eqPurityQuantum}
\end{equation}
For a qubit ($n=2$), this quantity has a nice geometrical interpretation in the Bloch ball: it is the \emph{squared length of the Bloch vector} which
corresponds to $\rho$. For all dimensions $n$,
\begin{equation}
   \p(\rho)=\left\{
      \begin{array}{cl}
         1 & \mbox{if }\rho\mbox{ is pure,}\\
         0 & \mbox{if }\rho\mbox{ is the maximally mixed state}.
      \end{array}
   \right.
   \label{eqPureMixed}
\end{equation}
The definition above applies to quantum theory, where states are density matrices on a Hilbert space. However, we can also consider
\emph{classical probability theory} (CPT) instead, where states are simply probability distributions, $p=(p_1,\ldots,p_n)$. In analogy to the
quantum definition, we set
\begin{equation}
   \p(p):=\frac n {n-1} \sum_{i=1}^n p_i^2 -\frac 1 {n-1}.
   \label{eqClassPur}
\end{equation}
In CPT, pure states are probability distributions like $p=(1,0,\ldots,0)$ (one unity, all others zero), and the maximally mixed state
is $p=\left(\frac 1 n,\ldots,\frac 1 n\right)$. Therefore, eq.~(\ref{eqPureMixed}) is still valid.

How can these definitions be naturally generalized to other possible probabilistic theories? (Readers who are not so interested in the
framework of general probabilistic theories may now safely proceed to the next subsection.) In the quantum case, the standard notion of purity can
be expressed as
\[
   \Tr(\rho^2)=\langle\rho,\rho\rangle,
\]
where $\langle X,Y\rangle:=\Tr(XY)$ is the Hilbert-Schmidt inner product on the real vector space of Hermitian matrices. This inner product
is very special: it is invariant with respect to unitary transformations $U$, that is, $\langle\mathcal{U}(X),\mathcal{U}(Y)\rangle=
\langle X,Y\rangle$, where we used the abbreviation $\mathcal{U}(X):=UXU^\dagger$.
This suggests the following strategy for defining purity in general probabilistic theories:
\emph{Find an inner product $\langle\cdot,\cdot\rangle$ on the state space
which is invariant with respect to all reversible transformations, and define the purity of a state $\omega$ as $\langle\omega,\omega\rangle$}.

To make this idea work, we have to be careful, though: even in quantum theory, the invariant inner product is not unique in the first place. This
is due to the fact that the space of Hermitian matrices $V$ decomposes into $V={\rm span}\{\mu\}\oplus \hat V$, where $\mu=\Id/n$ is
the maximally mixed state, and $\hat V$ is the subspace of traceless Hermitian matrices. These two subspaces are both invariant with
respect to unitaries. Thus, group representation theory~\cite{BarrySimon} tells us that there are infinitely many invariant inner products.

We can fix this problem by subtracting away the maximally mixed state $\mu$: if $\rho$ is a density matrix, we define the corresponding
``Bloch vector'' $\hat \rho:=\rho-\mu$. This is an element of the traceless Hermitian matrices $\hat V$, and that subspace cannot be further decomposed
into invariant subspaces. Thus, there is a unique inner product $\langle\cdot,\cdot\rangle$ (up to a constant factor) on $\hat V$, and we define
$\p(\rho):=\langle \hat\rho,\hat\rho\rangle$, rescaling the inner product such that $\p(\rho)=1$ for pure states $\rho$. It turns out that
$\langle X,Y\rangle=n/(n-1) \Tr(XY)$ for $X,Y\in\hat V$, and so this definition agrees with eq.~(\ref{eqPurityQuantum}).

In Section~\ref{SecMath}, we apply exactly the same construction to define purity in general probabilistic theories (cf.\ Definition~\ref{DefPurity}),
under some assumptions on the probabilistic theory which are necessary to get a useful definition. The resulting purity notion will,
in particular, still satisfy eq.~(\ref{eqPureMixed}).

\subsection{State space dimension $K$ and capacity $N$}
\label{SubsecKN}
For every state space, we denote by $K$ the \emph{number of real parameters required to describe an unnormalized, mixed state},
whereas $N$ denotes the \emph{maximal number of (normalized) states that can be perfectly distinguished in a single measurement}.
These quantities were to our knowledge first introduced by Wootters and Hardy~\cite{Wootters86,Hardy01}.

As a simple example, consider a single quantum bit (qubit). Arbitrary mixed states of a qubit are described by density matrices
$\rho=\left(\begin{array}{cc} w & y+iz \\ y-iz & x \end{array}\right)$, where normalization $\Tr\rho=1$ demands that $w+x=1$.
Since we are interested in \emph{unnormalized} states, we may drop this condition. As a result, unnormalized states $\rho$
are described by four real parameters $w,x,y,z$. (Positivity of the matrix adds additional constraints in the form of inequalities,
but the set of matrices fulfilling these conditions is still four-dimensional.) That is, we have $K=4$.
On the other hand, if we want to distinguish two states $\rho,\sigma$
perfectly in a single-shot measurement, they must be orthogonal. Since there are only two orthogonal states on $\C^2$, the \emph{capacity}
if the qubit state space is $N=2$.

For all state spaces of quantum theory, capacity $N$ equals the Hilbert space dimension (i.e.\ the states live on $\C^N$), and we have
the relation $K=N^2$.

In classical probability theory, the state space with $N$ perfectly distinguishable configurations consists of the probability
distributions $p=(p_1,\ldots,p_N)$ with $\sum_i p_i=1$. Dropping normalization, these are $N$ real parameters $p_1,\ldots,p_N$ to specify
a state. That is, classical state spaces have $K=N$, in contrast to quantum theory.

For other general probabilistic state spaces, state space dimension $K$ and capacity $N$ can basically be arbitrary natural numbers,
only the relation $K\geq N$ is always true. We give a rigorous mathematical definition of both quantities in Section~\ref{SecMath}.

\subsection{Unifying typical entanglement and coin tossing}
\label{SubsecUnify}
We will now show that eq.~(\ref{eqMain}) describes both typical entanglement of random pure quantum states
and classical coin tossing at the same time.
This will be demonstrated by considering three special cases of eq.~(\ref{eqMain}).

{\bf Random pure quantum states.}
Suppose we draw a pure state $\omega^{AB}$ on $A\otimes B$ at random, where $A$ and $B$ are Hilbert spaces of dimensions
$N_A$ and $N_B$ respectively. Recalling eq.~(\ref{eqPurityQuantum}) and~(\ref{eqPureMixed}) and the fact that $K=N^2$
in quantum theory, our main formula~(\ref{eqMain}) yields
\[
   \mathbb{E}_\omega \p(\omega^A)=\mathbb{E}_\omega\left(\frac {N_A}{N_A-1} \Tr\left[(\omega^A)^2\right] -\frac 1 {N_A-1}\right)
   =\frac{N_A^2-1}{N_A^2 N_B^2 -1}\cdot \frac{N_A N_B-1}{N_A -1}=\frac{N_A+1}{N_A N_B+1}\stackrel{N_B\to\infty}\sim \frac 1 {N_B}.
\]
Recall that $\p(\omega^A)$ is one if and only if $\omega^A$ is pure, and it is zero if and only if $\omega^A$ is the maximally
mixed state. Now if the ``bath`` $B$ becomes large, we see that the expected purity of the local reduced state on $A$ gets closer and
closer to zero, so that $\omega^A$ gets close to maximally mixed.
This expresses the fact that \emph{random pure quantum states are typically almost maximally entangled}, if the
bipartition is taken with respect to a small subsystem.\footnote{It is interesting to note that the set of stabilizer states (including their convex combinations as density matrices) shares the
values of $N$ and $K$ with usual quantum theory. As we show later, this set of states satisfies all the conditions for eq.~(\ref{eqMain}) to hold.
Thus, this equation gives the same amount of expected local purity as for the full set of quantum states. This was already observed in~\cite{SmithL05}, and is a consequence of the fact that the Clifford group constitutes a $2$-design.}

By \emph{typicality}, at this point, we mean something very specific. Suppose we want to generate a random state $\omega^{AB}$
with fixed purity $\p(\omega^{AB})=:\p_0$ (in this case $\p_0=1$, since we are interested in random \emph{pure} states).
We do this by choosing a fixed state $\varphi^{AB}$ with purity $\p(\varphi^{AB})=\p_0$, and then apply a random reversible transformation
$T$ to it, getting $\omega^{AB}:=T\varphi^{AB}$. The transformation $T$ is picked according to the invariant measure (Haar measure)
on the group of reversible transformations.
In the quantum case, the Haar measure on unitaries is also called the unitary circular ensemble. See~\cite{Diaconis05} for an explicit
recipe for how to pick unitaries numerically in this manner.

So far, our formula only expresses the \emph{expectation value} of the local purity $\p(\omega^A)$.
To call this value the {\em typical} value one needs to show that the distribution is peaked around the mean. Intuitively this must be the case if the expected value is close to the minimum allowed, as that could only occur if almost all of the distribution is concentrated close to the minimum. A simple way to see that this is indeed the case is to apply Markov's inequality~\cite{Billingsley}, which in this case reads (for $x>1$)
\[
   P\left\{\p(\omega^A)\geq \frac 1 x\right\} \leq x\cdot \mathbb{E}_\omega \p(\omega^A)\stackrel{N_B\to\infty}\approx \frac x {N_B}.
\]
This shows that if the mean is small, the probability of $\p(\omega^A)$ deviating from it must be small. Stronger results of this kind
can be obtained from measure concentration theorems on Lie groups~\cite{MilmanSchechtman}, but we will not pursue this approach
further in this paper.

\bigskip
{\bf Random pure classical states.} What if we draw random pure bipartite states in classical probability theory?
In this case, purity is defined by
eq.~(\ref{eqClassPur}), and, as discussed above, state space dimension and capacity are equal: $K=N$. Thus, our main formula~(\ref{eqMain})
yields for the local marginal $\omega^A=(\omega^A_1,\ldots,\omega^A_{N_A})$ the result
\[
   \mathbb{E}_\omega \p(\omega^A)=\mathbb{E}_\omega\left(\frac{N_A}{N_A-1}\sum_{i=1}^{N_A} (\omega^A_i)^2 -\frac 1 {N_A-1}\right)
   =\frac{K_A-1}{K_A K_B-1}\cdot\frac{N_A N_B-1}{N_A-1}=1.
\]
All the terms cancel, and we get that the expected local reduced purity equals unity. Since $1$ is the maximal value, this is only possible if in fact
$\p(\omega^A)=1$ for \emph{all} pure states $\omega^{AB}$. In other words: all pure bipartite states have pure marginals. This expresses
the simple fact that \emph{there are no entangled states in classical probability theory} -- all pure bipartite states are product states.

Before turning to the more interesting example of classical coin tossing, we briefly discuss how classical probability distributions
$\omega^{AB}$ of fixed purity $\p(\omega^{AB})=:\p_0$ are drawn at random (in this example, so far, we have the case $\p_0=1$).
In analogy to the quantum case, we start with an arbitrary fixed
bipartite probability distribution $\varphi^{AB}$ with purity $\p(\varphi^{AB})=\p_0$. In classical probability theory, the reversible transformations are
the \emph{permutations}, that is, doubly stochastic matrices containing only ones and zeroes. Now the state $\omega^{AB}$ is defined
as $\omega^{AB}:=T\varphi^{AB}$, where $T$ is a random permutation.

\bigskip
{\bf Classical coin tossing.} We can use identity~(\ref{eqMain}) to describe the process of coin tossing in classical probability
theory. Suppose we start with a coin (that is, a classical bit) whose value (``heads'' or ``tails'' -- say, heads) is perfectly known to us. In this case, the
(pure) state of the coin is a probability distribution $\varphi^A=(1,0)$, where $1$ is the probability of heads and $0$ the probability of tails.
However, the environment $B$ is not known to us -- it is in some mixed state $\varphi^B$. The total state (coin and environment) is thus
in a mixed state $\varphi^{AB}:=\varphi^A\otimes\varphi^B$, with some purity $\p(\varphi^{AB})=:\p_0<1$.

Now we toss the coin -- that is, we flip it in an uncontrolled way which makes it interact with the environment. It makes sense to model
this process as a random reversible transformation (permutation) $T$ of the global system $AB$. In the end, we capture the coin, cover it with
our hand (so that we cannot see what side is up) and disregard the environment. The state of the coin is then $\omega^A$, the marginal
corresponding to the global state $\omega^{AB}:=T\varphi^{AB}$. We expect that the coin's state should be mixed. In fact,
\[
   \mathbb{E}_\omega \p(\omega^A)=\frac{K_A-1}{K_A K_B-1}\cdot\frac{N_A N_B-1}{N_A-1}\cdot\p(\omega^{AB})=\p(\omega^{AB})=\p_0,
\]
where the same cancellation as in the previous example applies. That is, \emph{our ignorance about the environment gets transferred
to the coin}, which is exactly what coin tossing is all about.

\subsection{Typical entanglement of symmetric and antisymmetric states}
\label{SubsecSymm}
{So far, our discussion only covered the case that a state is drawn randomly from the set of \emph{all} states (subject
to fixed purity). However, there are situations -- particularly in thermodynamics, as we discuss in the next subsection -- where
one would like to draw random states subject to additional constraints.

This generalization is treated in Subsection~\ref{SubsecSym}, where we compute the expected subsystem purity for random
states that satisfy certain symmetry constraints (Theorem~\ref{TheFace}). In the special case of quantum theory, this gives
the typical entanglement for subspaces $S\subseteq AB$ that have the following symmetry property: \emph{For every unitary
$U$ on $A$, there is a unitary $U'$ on $B$ such that $U\otimes U'$ preserves the subspace $S$.} An explicit formula for
the expected purity of the reduced state is given in Theorem~\ref{TheQFace}.

Here, we apply this result to compute the typical amount of entanglement in pure states of symmetric and antisymmetric subspaces:
they are both $U\otimes U$-invariant.
As usual, for a Hilbert space ${\cal H}$, the symmetric subspace ${\cal H}\vee{\cal H}$ resp.\
antisymmetric subspace ${\cal H}\wedge{\cal H}$ are defined as those vectors $|\psi\rangle$ with
$\pi|\psi\rangle=|\psi\rangle$ resp.\ $\pi|\psi\rangle=-|\psi\rangle$, where $\pi$ is the unitary that swaps the two particles.
For three and more particles, the totally (anti)symmetric subspace is defined as the set of vectors that satisfy this
equation for all pairs of particles simultaneously.}
Investigating this case is motivated by the importance of identical bosons and fermions which, by the symmetrisation postulate, have symmetric
and antisymmetric joint states respectively~\cite{Peres02,CohenTannoudjiDL06}.

States such as antisymmetric fermionic states are clearly entangled in the mathematical sense, but we note that they can only be termed entangled in the operational sense
under some additional assumptions.  
Standard entanglement theory implicitly assumes that different systems corresponding to different tensor factors can be operationally
distinguished, which is in general not true for bosons and fermions. However,
whilst e.g.\ two electrons are always indistinguishable, they can in fact be treated as distinguishable if they are localized in two separate spatial locations~\cite{Peres02,CohenTannoudjiDL06}.
This fact gives rise to a natural scenario where antisymmetric states appear that are entangled in the operational sense:
If two such localized electrons had previously shared the same spatial part of the wavefunction, their internal degrees of freedom (spin) would have been antisymmetric, and remain so unless altered. After separating the two electrons, one would have obtained standard (not just mathematical) entanglement between them, having  arisen due to the antisymmetry requirement on their joint state (see eg.~\cite{GittingsF02,Vedral03, Cavalcanti07}). One may, for concreteness, think about our calculations in this section with such a scenario in mind.  
\begin{theorem}
\label{ThePuritySymm}
Consider the symmetric and antisymmetric subspaces $S_\pm$ on two $n$-level quantum systems $A=B=\C^n$, i.e.\
$S_+=\C^n\vee \C^n$ and $S_-=\C^n\wedge\C^n$. If $\omega_\pm\in S_\pm$ is a random pure quantum state,
then the expected local purity is
\[
  \mathbb{E}_{\omega_{\pm}} \Tr\left[\left(\omega_\pm^A\right)^2\right] = \frac{2(n\pm 1)}{n^2 \pm n +2}.
\]
Moreover, drawing a random mixed state of fixed purity $\Tr(\omega_\pm^2)$ from the corresponding subspace, the expected local purity is described
by the same equation, only the factor $2$ in the numerator has to be replaced by $1+\Tr(\omega_\pm^2)$.
\end{theorem}
\proof
We use Theorem~\ref{TheQFace} from Subsection~\ref{SubsecSym}: this theorem is applicable because the symmetric
and antisymmetric subspace are both invariant with respect to transformations of the form $U\otimes U$. In the following,
we sketch the proof for the symmetric subspace $S:=S_+$; the proof for the antisymmetric case is completely analogous.

According to the notation of Theorem~\ref{TheQFace}, we have $N_A=n$ and $N_S=\dim(S_+)=n(n+1)/2$. Since the case
$n=1$ is trivial, we may assume that $n\geq 2$. Denote orthonormal basis vectors of $A=\C^n$ by $|1\rangle,|2\rangle,\ldots,|n\rangle$.
We may choose the matrix $E_A$ as $E_A:=\frac 1 {\sqrt{2}} |1\rangle\langle 1|-\frac 1 {\sqrt{2}}|2\rangle\langle 2|$; it
satisfies $\Tr E_A=0$ and $\Tr E_A^2=1$ as required. An orthonormal basis of $S$ consists of the vectors $|ii\rangle$ with $1\leq i\leq n$,
and $\frac 1 {\sqrt{2}}\left( |ij\rangle+|ji\rangle\right)$ for $i<j$. This allows us to write the projector $\pi$ onto $S$ as
\[
  \pi=\sum_{i=1}^n |ii\rangle\langle ii| + \frac 1 2 \sum_{i<j} \left(|ij\rangle+|ji\rangle\right)\left(\langle ij|+\langle ji|\right)
  =\frac 1 2 \sum_{i,j} |ij\rangle\langle ij| + \frac 1 2 \sum_{i,j} |ij\rangle\langle ji|.
\]
Using these expressions, the calculation of $\Tr\left[ (\pi(E_A\otimes\Id_B)\pi)^2\right]$ is lengthy but straightforward.
The result is that this expression equals $n/4+1/2$. Substituting this into Theorem~\ref{TheQFace} proves the claim.
\qed

Since Theorem~\ref{TheQFace} (which has been used to prove this result) is applicable in more general situations, there exist
several possibilities to generalize the theorem above. For example, consider the totally symmetric or totally antisymmetric subspace
on $N$ qudits, $S_+:=\C^n\vee \C^n\vee\ldots\vee\C^n$, and $S_-:=\C^n\wedge \C^n\wedge\ldots\wedge\C^n$, both as
subspaces of $AB=(\C^n)^{\otimes N}$. Consider the $1$-versus-$(N-1)$-qudits cut, i.e.\ $A=\C^n$ and $B=(\C^n)^{\otimes (N-1)}$.
Then the situation satisfies the conditions of Theorem~\ref{TheQFace}: for every unitary $U$ on $A$, there is a unitary $U'$ on $B$
such that $U\otimes U' S_\pm = S_\pm$, namely $U':=U^{\otimes (N-1)}$. Thus, Theorem~\ref{TheQFace} can be used to compute
the expected local purity for this cut. (We do not pursue this calculation here.)

It is clear that the result of Theorem~\ref{ThePuritySymm} above can be proven in principle without the machinery of this paper,
purely within quantum mechanics. However, we think it is important to have it derived within the framework of
general probabilistic theories, showing the power and flexibility of this framework. Our general proof, as given in Section~\ref{SecMath},
is very geometrical in flavour; it treats the set of quantum states as a convex set, with the (anti)symmetric subspace as a face.
Thereby, it shows very clearly what geometric properties of the quantum state space are important for the result to hold.

Apart from the generalization to other theories that one obtains for free, this proof method also clarifies some aspects of the
quantum result. For example, it shows why some further generalizations of the above result will need considerable further effort, such as computing
the average local purity for the $2$-versus-$(N-2)$-cut on the totally symmetric subspace. If Alice holds two qudits, she can 
locally perform unitaries of the form $U\otimes U$. If, for any such unitary, the map $U':=U^{\otimes (N-2)}$ is applied on
Bob's part of the state, then the totally symmetric subspace stays invariant. Since the group of unitaries $U\otimes U$ acts
irreducibly on Alice's subspace $\C^n\wedge\C^n$, the situations seems fine at first, and one might guess that Theorem~\ref{TheQFace} is easily generalized to this situation.

But this turns out to be wrong: the important property is that Alice's unitaries should act \emph{irreducibly
on her convex set of states}, which is in general not the case. Instead, the group action $\rho\mapsto U\otimes U \rho U^\dagger
\otimes U^\dagger$ is reducible on the space of traceless Hermitian matrices over $\C^2\vee \C^2$,
and Alice's (Bloch) state space decomposes into invariant subspaces. This shows that the relevant question is not whether
the group of $U\otimes U$ acts irreducibly, but whether it is a $2$-design. In this case, the answer is negative.

\subsection{Statistical physics and the second law}
\label{SubsecStatPhys}
Our result can also be used to generalize an approach to thermodynamics which has recently attracted a lot of
attention~\cite{HaydenP07,CalsamigliaHDB07,GemmerOM01,PopescuSW06,Lloyd88,GemmerMM04,LubkinL93,MuellerGE10}.
This approach is based on the fact that most pure quantum states are almost maximally entangled, in the sense described earlier
in this paper.

The main idea, as developed for example in~\cite{PopescuSW06}, can be stated as follows. We divide the universe's Hilbert
space ${\cal H}$ into a small ``system'' and a large ``environment'', ${\cal H}={\cal H}_S\otimes{\cal H}_E$. In many cases,
the state of the universe is constrained to be an element of some subspace ${\cal H}_R\subseteq {\cal H}$, which might be,
for example, a subspace corresponding to a narrow window of energies. The maximally mixed state on ${\cal H}_R$ is called
the ``equiprobable state``$\epsilon_R$. The actual state of the universe is then assumed to be some unknown pure state $|\psi\rangle$
from ${\cal H}_R$.

At first, it seems as if the exact form of the actual state $|\psi\rangle\in{\cal H}_R$ would have profound consequences, and that
very little can be said about the reduced state on the small subsystem, $\psi_S=\Tr_E |\psi\rangle\langle\psi|$. But this turns out
to be wrong: in fact, ``most'' states $|\psi\rangle$ look very alike on the small subsystem. That is, $\psi_S\approx \Tr_E \epsilon_R$
with high probability for randomly chosen $|\psi\rangle$. This can be formulated as follows:

\begin{quotation}
{\bf Principle of Apparently Equal a priori Probability~\cite{PopescuSW06}:}
For almost every pure state of the universe, the
state of a sufficiently small subsystem is approximately
the same as if the universe were in the equiprobable state
$\epsilon_R$. In other words, almost every pure state of the universe
is locally (i.e.\ on the system) indistinguishable from
$\epsilon_R$.
\end{quotation}
This principle is then used to justify the `equal a priory probability' assumption, an assumption of Statistical Physics which is used in the derivation of many major results in that field.

Our results can be interpreted in a similar manner. First, consider the simple case where we have a small system, $A$ (not
necessarily quantum), coupled to a large
bath, $B$, and where all global states are in principle possible. In the quantum situation, this corresponds to the special case where
${\cal H}_R={\cal H}$. If we have a random state $\omega^{AB}$ on $AB$ with purity $\p(\omega^{AB})$, and if the conditions of
Theorem~\ref{TheMainMain} are satisfied, we have for $N_B\gg K_A$
\[
   \mathbb{E}_\omega \p(\omega^A)= \frac{K_A-1}{K_A K_B-1}\cdot\frac{N_A N_B-1}{N_A-1}\cdot\p(\omega^{AB})
   \approx \frac{N_A(K_A-1)}{K_A(N_A-1)}\cdot\p(\omega^{AB})\cdot\frac{N_B}{K_B}.
\]
If this is very small, then the state of the small subsystem is very close to maximally mixed. As discussed in Subsection~\ref{SubsecUnify}), the
Markov inequality (or more powerful measure concentration inequalities) tell us that the expectation value is then also the typical value. In this case,
the ``Principle of Apparently Equal a priori Probability'' is satisfied in our more general setting.

It remains to see under what conditions this expectation value is actually close to zero. There are two possibilities how this may happen:
\begin{itemize}
\item We might have a random pure state, i.e.\ $\p(\omega^{AB})=1$, but $N_B/K_B$ might tend to zero with increasing size of the bath $B$.
This is exactly what happens in quantum theory, where $N_B$ is the bath's Hilbert space dimension, and $K_B=N_B^2$. The interpretation
is that ``most'' pure bipartite states are almost maximally entangled, such that the local reduced state looks close to maximally mixed.

It is interesting to see that the same phenomenon may appear in general probabilistic theories beyond quantum theory, and may
in fact be stronger:  there are natural possible classes of theories~\cite{Wootters86,Hardy01} where $K_B=N_B^r$ for some
integer $r\in\mathbb{N}$. While we have $r=2$ for quantum theory, other theories with $r\geq 3$ would have even ``stronger-than-quantum
randomization'': they would have $N_B/K_B=N_B^{1-r}$, turning faster to zero than the quantum value.
\item On the other hand, we may have $K_B\approx N_B$, or equality as in the case of classical probability theory. Then,
we could still have randomization if $\p(\omega^{AB})$ tended to zero with increasing size of $B$, in situations where
it makes sense to model the global state as a random \emph{mixed} state.

A situation like this is given in classical coin tossing, as discussed in Subsection~\ref{SubsecUnify}: there, $\p(\omega^{AB})$
describes the purity of the unknown global initial state, before the coin is tossed in a random, reversible way. Larger
environment usually amounts to less knowledge about its details, which means smaller purity $\p(\omega^{AB})$.

One may argue that a situation like this is also encountered in natural systems of classical statistical mechanics, if a small
finite system is reversibly coupled to a large, unknown environment.
\end{itemize}

In the quantum situation, the principle above is formulated for the more general case that ${\cal H}_R$ is a proper subspace of ${\cal H}$,
and not all of the global Hilbert space. The analogue of this situation in general probabilistic theories would be to have the set of allowed
states restricted to some face of the global state space $AB$. Our results do not directly address this situation in full generality, but
Theorem~\ref{TheFace} covers the special case of a $GG'$-invariant face $\mathbb{F}$. Even though the resulting formula is not as transparent as
the one above, it shows that the amount of randomization is also very strong if the face's dimension $K_\mathbb{F}$ increases with
the bath $B$: since projections are contractions, we have
\begin{eqnarray*}
   \mathbb{E}_\omega^{\mathbb{F}} \p(\omega^A)&=&\left\|\pi_{\bar \F} (X^A\otimes u^B)^\wedge\right\|_2^2
   \cdot\frac{K_A-1}{K_\F -1}\cdot\left(\strut\p\left(\omega^{AB}\right)-\p(\mu_\F)\right)
   \leq \left\|\left(X^A\otimes u^B\right)^\wedge\right\|_2^2 \cdot\frac{K_A-1}{K_\mathbb{F}-1} \, \p\left(\omega^{AB}\right)\\
   &=& \frac{K_A-1}{K_\F -1}\cdot\frac{\p\left(\omega^{AB}\right)}{\p\left(\varphi^A\otimes\mu^B\right)}
   =\frac{K_A-1}{K_\F -1}\cdot \frac{N_A N_B-1}{N_A-1}\cdot\p(\omega^{AB})
\end{eqnarray*}
whenever the conditions of Theorem~\ref{TheMainMain} are satisfied (we have also used Lemma~\ref{LemGlobalPauli}). If $N_B/K_\F$
tends to zero with increasing size of the bath (as is the case for symmetric and antisymmetric subspaces in quantum theory), the
Principle of Apparently Equal a priori Probability remains valid.

A speculative, but interesting application of this result in the post-quantum case could be in black hole thermodynamics.
The results on typical entanglement have already been discussed in the context of black hole entropy~\cite{HsuReeb}, and quantum information
analysis has been applied to learn more about the black hole information paradox~\cite{HaydenP07,SmolinOppenheim,PageBH,PreskillBH}.
Since no fully complete and unique theory of quantum gravity is available yet, many parts of black hole thermodynamics are subject to speculation.
Vice versa, the assumption that the laws of thermodynamics are valid for black holes is used to obtain information on properties of the
possible underlying theory of quantum gravity.

One may speculate that a possible theory of quantum gravity might not only involve a modification of the usual concepts of geometry and gravity,
but also of quantum theory itself. It is possible that quantum theory is only an approximation to a different kind of deeper probabilistic theory, similarly
as classical probability is only an approximation to quantum theory.
The principle of equal a priori probability is closely linked to the second law, and one may
view our results as a first step towards formulating the second law as a kind of meta-theorem that does not depend on the details of the theory and may thus apply to post-quantum theories.

As further motivation for research in this direction we note that there is a striking historical precedent where
assuming the persistence of the second law
helped to discover new physics. Planck~\cite{Planck1900} arrived at energy quantization
(energy $\varepsilon = h \nu$, where $\nu$ is a frequency and $h$ his constant), by implicitly assuming certain thermodynamical entropic relations would still be valid after the quantization of energy.

\subsection{Simple proof of the quantum case}
\label{SubsecQProof}
We now give a comparatively simple derivation of the value of typical purity for the quantum case. The proof simplifies and generalizes the
proof for the case of globally pure quantum states in~\cite{OliveiraDP07, DahlstenOP08}. Moreover, it gives some intuition on the
necessary notions and ingredients for the general proof in Section~\ref{SecMath}.

Firstly we note that the local purity is directly related to how well one can predict measurements of local outcomes. 
Phrased in these terms we wish to show that local measurements (of the form $g_A\otimes \Id_B$ for some $g_A\neq\Id_A$) tend to be highly unpredictable. We shall accordingly represent the state in a way which makes it clear to what extent local measurements are defined. We will use a
nice way of linking the Heisenberg and Schr\"odinger pictures, which is to expand the density matrix in terms of elements $g$ of the Pauli group $\{X,Y,Z,\id\}^{\otimes n}$:
\[
\rho=\sum_i \xi_i g_i.
\]
The sum contains $4^n$ terms, which we label from $i=0$ to $i=4^n-1$, such that $g_0=\id$.
The coefficients $\xi_i$ are directly related to the expectation values of the corresponding Pauli element via 
\begin{equation}
\label{eq:xi}
\left<g_i\right>=\Tr(\rho g_i)=\Tr(\id\xi_i)=2^n\xi_i.
\end{equation}
 
In what follows we use the above representation to derive the expected purity value. We write the formula in a way that highlights that the ratio of the local to the total purity is proportional to the ratio of the number of local versus global observables. The intuition here is as follows. There is a certain limited amount of purity/predictability about the state, and this gets associated with observables picked at random (not independently) by the
random unitary. If most observables are global, this predictability
is then likely to be associated with global observables, while the remaining local ones become unpredictable.

\begin{lemma} Consider any quantum state $\varphi$ on $n=n_A+n_B$ qubits, with fixed purity $\Tr(\varphi^2)$.
Apply a random unitary $U$ to it, i.e.\ $\rho:=U\varphi U^\dagger$.
Then the expected local purity on subsystem $A$ is given by
\[
   \frac{\mathbb{E}_U \Tr\left[(\rho^A)^2\right] - 2^{-n_A}}{\Tr(\varphi^2)-2^{-n}}
   =2^{n_B}\cdot\frac{K_A-1}{K_{AB}-1},
\]
where $K_A=4^{n_A}$ and $K_{AB}=4^{n}$ quantify the number of local (i.e.\ Paulis of the form $g_A\otimes\Id_B$)
and global degrees of freedom (all other Paulis), respectively. Note that $2^{-n_A}$ and $2^{-n}$ are the \emph{minimal} possible
values of the purity of any quantum state on $A$ resp.\ $AB$.
\end{lemma}
\begin{proof}
Note that $\rho^A=\Tr_B(\rho)=\sum_{\text{Paulis }g_i \text{ with } \id \text{on B}}\xi_i \Tr_B(g_i)$,
and there are $4^{n_A}$ such elements. This shows that $(\rho^A)^2=(2^{n_B})^2 \sum_{i=0}^{4^{n_A}-1} \xi_i^2 \id_A$,
where $i$ is the label of the Pauli operator $g_A\otimes\Id_B$. Consequently
\begin{eqnarray*}
\mathbb{E}_U \Tr\left[(\rho^A)^2 \right] &=& 2^{n+n_B}\sum_{i=0}^{4^{n_A}-1}\mean_U \xi_i^2
= 2^{n+n_B}\left[ \mean_U \xi_0^2+\sum_{i=1}^{4^{n_A}-1}\mean_U \xi_i^2 \right] 
= 2^{n+n_B}\left[ 2^{-2n}+\left(4^{n_A}-1\right)\mean_U \xi_i^2  \right].
\end{eqnarray*}
We have used the fact that $\Tr(\rho)=1=2^n\xi_0\Rightarrow \mean_U \xi_0^2=2^{-2n}$.
Now consider two elements $g_i,g_j$ with $i,j\neq 0$. Those elements are connected by some unitary
operation $V$, i.e. $g_j=V g_i V^\dagger$. Thus,
\[
   \mathbb{E}_U \xi_j^2 =\mathbb{E}_U \Tr\left[(\varphi U^\dagger g_j U)^2\right] 
   =\mathbb{E}_U \Tr\left[(\varphi U^\dagger V g_i V^\dagger U)^2\right] = \mathbb{E}_U \xi_i^2
\]
due to the unitary invariance of the Haar measure. Now we exploit the fact that $\Tr(\varphi^2)=\Tr(\rho^2)=
2^n \left(\sum_{i=1}^{4^n-1} \xi_i^2 + 2^{-2n}\right)$. Taking the expectation value of this expression gives
$\displaystyle\mathbb{E}_U \xi_i^2=2^{-n}\cdot
\frac{\Tr(\varphi^2)-2^{-n}}{4^n-1}$. This can be substituted into the expression above, proving the statement of the lemma.
\end{proof}

Some remarks:
\begin{itemize}
\item[1.] We see that
$\mean_U \Tr\left[(\rho^A)^2 \right]\approx 2^{\left(n_B-n\right)}=2^{-n_A}$ when $n\gg n_A \gg 1$. This is the minimum value it can take.  
\item[2.] The purity condition $\Tr(\rho^2)=1$ enforces the uncertainty principle, forcing most of the $\xi_i^2$ to be small.
\item[3.] Apart from the purity restriction, the observer is constrained by having access to only the $K_A-1=4^{n_A}-1$ local observables
out of a total of $K_{AB}-1=4^{n}-1$. This ratio appears directly in the statement of the lemma.
\end{itemize}
Perhaps surprisingly, the proof above can also be adapted to classical probability theory (with $n_A+n_B$ classical bits). The only difference
in the result will be that $K_A$ and $K_{AB}$ have to be replaced by $K_A^{\rm CPT}=2^{n_A}$ and $K_{AB}^{\rm CPT}=2^n$,
in agreement with Subsection~\ref{SubsecKN}. Remark 2.\ above suggests that this result is related to the absence of uncertainty
in classical pure states.

The proof above also illustrates some main ideas for the derivation of the general probabilistic result in Section~\ref{SecMath}.
A useful insight above was to consider the linear maps $\xi_i\equiv \xi_i(\rho)$, and to see that purity can in general be expressed as a sum over
$\xi_i(\rho)^2$. More specifically, the \emph{local reduced} state's purity can be expressed as a sum over $\xi_i(\rho)^2$ for a
certain type of $\xi_i$'s, namely those which act locally: $\xi_i(\rho)=\Tr(\rho(g\otimes\Id))$. Since they are all connected by reversible
transformations, they all have the same Haar expectation value.

The general case will use a very similar construction, where the sum is replaced by an integral over the group, and the map $\xi_i$ is replaced by a general ``Pauli map'' (cf.\ Lemma~\ref{LemFormulaPauliPurity}). In analogy to the quantum case, it turns out that the local reduced state's purity can be expressed as an integral over a certain type of Pauli map, namely one which acts locally (cf.\ Lemma~\ref{LemGlobalPauli} and
the proof of Theorem~\ref{TheMain1}.) Again, due to invariance with respect to reversible transformations, all these maps in the integral
have the same Haar expectation value, giving rise to the proof of our main result.

\section{Mathematical framework and proofs}
\label{SecMath}
\subsection{General probabilistic theories and the Bloch representation}
\label{SubsecGenBloch}
We work in the framework of \emph{general probabilistic theories}, a natural mathematical framework which describes basic operational laboratory situations like preparations, transformations, and measurements. Quantum theory can be described within the framework, as well as classical probability theory and a large class of possible generalizations. For an introduction to this framework, and in particular for the physical motivation, see e.g.~\cite{Hardy01, Barrett07,BarnumWilce08}. Our notation is particularly close to~\cite{BarnumWilce08} and~\cite{Hardy01}.

A \emph{state space} is a tuple $(A,A_+,u^A)$, where $A$ is a real vector space of finite dimension $K_A$ (we will not consider infinite-dimensional state spaces in this paper), and $A_+\subset A$ is a proper cone (that is, a closed, convex cone
of full dimension which does not contain lines). It can be interpreted as the set of unnormalized states.
$u^A$ is a linear functional which is strictly positive on $A_+\setminus\{0\}$ and
is called the \emph{order unit} of $A$.
The set of points $\omega\in A_+$ with $u^A(\omega)=1$ is called the set of (normalized)
\emph{states} and is denoted $\Omega_A$. It follows that $A_+=\bigcup_{\lambda\geq 0} \lambda\Omega_A$ and
that $\Omega_A$ is a compact convex $(K_A-1)$-dimensional set. Its extremal points are called \emph{pure states},
the others are \emph{mixed states}. Instead of the full tuple, we will usually just call $A$ the ``state space''.

A linear invertible map $T:A\to A$ is called a \emph{symmetry} if $T(A_+)=A_+$ and $u^A\circ T=u^A$. That is,
symmetries $T$ map the set of normalized states $\Omega_A$ bijectively into itself. The example of a qubit shows
that not all symmetries of a state space have to be allowed transformations: reflections in the Bloch ball
are symmetries, but are not physically allowed since they do not correspond to completely positive maps.
Thus, in order to define reversible dynamics on a state space, we also have to specify a group
$\G_A$ of (allowed) \emph{reversible transformations}.
For the sake of generality, we allow arbitrary choices of $\G_A$, as long as $\G_A$ is compact and contains
only symmetries.\footnote{The physical motivation for postulating \emph{compact} groups $\G_A$ is as follows.
First, $\G_A$ must be bounded (in the topology induced by its action on $\Omega_A$) due to the compactness
of $\Omega_A$. Then, suppose we have a sequence
of transformations $(T_n)_{n\in\N}\subset \G_A$ such that $\lim_{n\to\infty}T_n=T$. Physically, this means that
we can apply the transformation $T$ to arbitrary accuracy. But this is anyway all that we can hope for in physics;
hence it makes sense to call $T$ a physically allowed reversible transformation, and include it in $\G_A$. Thus,
from a physical point of view, it makes sense to postulate that $\G_A$ must be closed.}
Then, a pair $(\mathbf{A},\G_A)$, where $\mathbf{A}$ is a state space (equivalently: a tuple $(A,A_+,u^A,\G_A)$) will be called
a \emph{dynamical state space}. Again, to save some ink, we will usually denote the dynamical state space simply by the
letter $A$ rather than by the full tuple.

One goal of this paper is to investigate properties of random pure states on general state spaces.
In order to have a meaningful mathematical notion of ``random states'', we need the following property:

\begin{definition}[Transitivity]
\label{DefTransitivity}
A dynamical state space $A$ is called \emph{transitive} if for every pair of pure states $\alpha,\omega\in\Omega_A$
there exists a reversible transformation $T\in\G_A$ such that $T\alpha=\omega$.
\end{definition}

Since $\G_A$ is compact, we have the notion of a Haar measure on that group~\cite{BarrySimon}. Thus, we can draw a pure state by applying a random reversible transformation $T\in\G_A$ to an arbitrary given pure state $\omega$. Transitivity is required so that the resulting distribution does not depend on the initial state $\omega$. The property of transitivity is thus a necessary mathematical prerequisite in order to have an unambiguous notion of ``random pure states''. 

It is also questionable whether reversible theories without transitivity re self-consistent in a specific physical sense. Imagine for example that given a product state there is no reversible way to transform it into a pure but correlated (and thus entangled) state. This is the case for the theory with PR-boxes known as boxworld~\cite{boxworld}. Then one can for example not model a measurement as a reversible correlating interaction between a memory system and the system in question. The possibility of modelling measurement interactions in that way seems to play a fundamental role for the self-consistency of quantum theory and statistical mechanics (cf.\ Maxwell's demon and Bennett's reversible measurements).
In accordance with the two justifications above we shall henceforth unless otherwise stated assume that we are dealing with transitive state spaces.

\begin{definition}[Maximally mixed state]
\label{DefMaxMix}
If $A$ is a transitive dynamical state space, let $\omega\in\Omega_A$ be an arbitrary pure state, and define
the \emph{maximally mixed state} $\mu^A$ on $A$ by
\[
   \mu^A:=\int_{G\in\G_A} G(\omega)\, dG.
\]
\end{definition}

It follows from transitivity that $\mu^A$ does not depend on the choice of $\omega$. Clearly, $T\mu^A=\mu^A$ for all
$T\in\G_A$, and $\mu^A$ is the unique state on $A$ with this invariance property. The space $A$ can be decomposed
into a direct sum
\[
   A=\hat A\oplus \R\mu^A,
\]
where $\R\mu^A$ denotes the one-dimensional subspace which is spanned by $\mu^A$, and $\hat A$ is defined as the set of
all vectors $a\in A$ with $u^A(a)=0$. 
\begin{definition} [Bloch vector]
Given any state $\omega\in \Omega_A$
(or, more generally, any point $\omega\in A$ with $u^A(\omega)=1$), we define its corresponding \emph{Bloch vector} $\hat\omega$ as

\[
   \hat\omega:=\omega-\mu^A\in \hat A.
\]
The set of all Bloch vectors $\hat\omega$ with $\omega\in\Omega_A$ will be called $\hat\Omega_A$.
\end{definition}
Note that convex combinations of states yield the corresponding convex combinations of the Bloch vectors:
$\left(\sum_i \lambda_i \omega_i\right)^\wedge=\sum_i \lambda_i \hat \omega_i$ if $\sum_i\lambda_i=1$.
Every reversible transformation $T\in\G_A$ leaves $\mu^A$ invariant. Thus, we have
\[
   (T\omega)^\wedge=T\omega-\mu^A=T(\omega-\mu^A)=T\hat\omega,
\]
and applying a transformation $T$ to a state is equivalent to applying it to the corresponding Bloch vector.

\subsection{Definition and properties of purity}
We would like to define a notion of \emph{purity} in generalized state spaces. In the quantum case, this is just
$\Tr(\rho^2)=\langle\rho,\rho\rangle$, where $\langle X,Y\rangle:=\Tr(XY)$ denotes the Hilbert-Schmidt inner product
on Hermitian matrices. This inner product is very special -- it is \emph{invariant} with respect to the reversible transformations
of quantum theory: $\langle UA U^\dagger,UBU^\dagger\rangle=\langle A,B\rangle$ for all unitaries $U$. Thus, it makes
sense to ask for the existence of an analogous inner product in more general theories.

\begin{lemma}
\label{LemInvariantInnerProduct}
Let $A$ be a transitive dynamical state space. Then the following statements are equivalent:
\begin{itemize}
\item There is a unique inner product $\langle \cdot,\cdot\rangle$ on $\hat A$ (up to constant multiples) such that all reversible
transformations $T\in \G_A$ are orthogonal~\footnote{A linear map $T$ is orthogonal with respect to an inner product
$\langle\cdot,\cdot\rangle$ if $\langle Tu,Tv\rangle=\langle u,v\rangle$ for all vectors $u$ and $v$.}.
\item $\G_A$ acts irreducibly on $\hat A$; that is, $\hat A$ does not contain any proper subspace which is invariant
under the action of all reversible transformations $ T\in\G_A$.
\end{itemize}
\end{lemma}
\proof
This is a standard use of the (real version of) Schur's Lemma, see Proposition VIII.2.3 in~\cite{BarrySimon}.
\qed

If $\G_A$ acts irreducibly on $\hat A$, we call the dynamical state space $A$ \emph{irreducible}. In other words: $A$ is irreducible
if and only if the only non-trivial subspaces which are invariant under all transformations of $\G_A$ are $\hat A$ and $\R\mu^A$.

Transitive dynamical state spaces are not automatically irreducible. As a simple example,
consider a state space $\hat\Omega_A$ which is a cylinder (as in Figure~\ref{fig_cylinder}) and where $\G_A$ contains all symmetries. The pure states
are the points on the two circles. By rotation and reflection, every pure state can be reversibly
mapped to every other, such that we have transitivity. However, it is not irreducible:
the symmetry axis and the plane orthogonal it, intersecting the cylinder's center, are
invariant subspaces.
 \begin{figure}[!hbt]
 \begin{center}
 \includegraphics[angle=0, width=4cm]{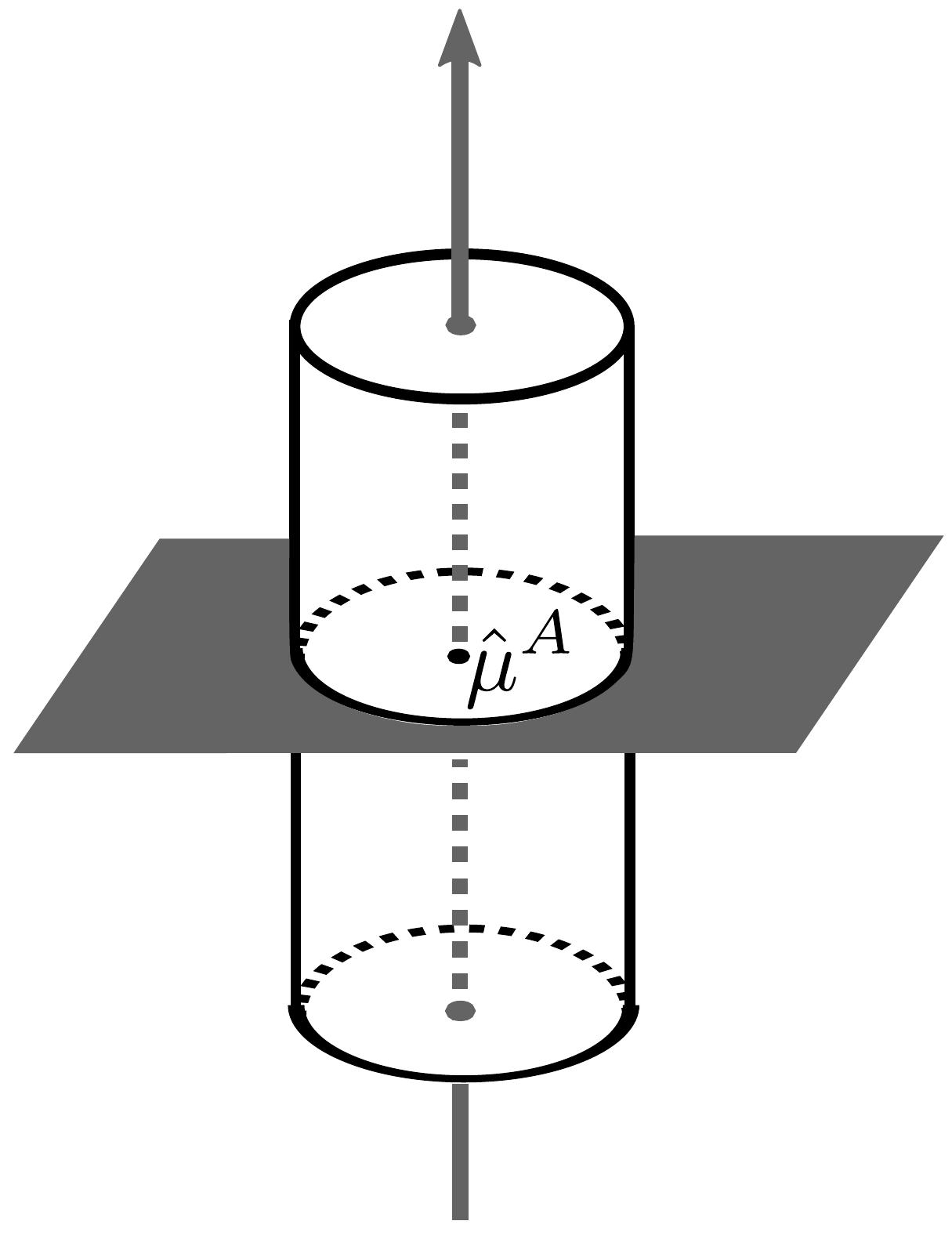}
 \caption{If $\hat\Omega_A$ is a cylinder, then the corresponding state space is transitive, but not irreducible.}
 \label{fig_cylinder}
 \end{center}
 \end{figure}

Now it is straightforward to introduce a generalized notion of purity.
\begin{definition}
\label{DefPurity}
Let $A$ be a transitive and irreducible state space, and let $\langle\cdot,\cdot\rangle$ be the unique
inner product on $\hat A$ such that all transformations are orthogonal and $\langle\hat\alpha,\hat\alpha\rangle=1$
for pure states $\alpha$. Then, the \emph{purity} $\p(\omega)$ of any state $\omega\in\Omega_A$ is defined as
the squared length of the corresponding Bloch vector, i.e.
\[
   \p(\omega):=\|\hat\omega\|^2 \equiv\langle\hat\omega,\hat\omega\rangle.
\]
\end{definition}

It is straightforward to deduce some useful properties that follow from this definition:
\begin{lemma}[Properties of Purity]
\label{LemPropPur}
Let $A$ be a transitive and irreducible dynamical state space, then
\begin{itemize}
\item[1.] $0\leq \p(\omega)\leq 1$ for all $\omega\in\Omega_A$,
\item[2.] $\p(\omega)=0$ if and only if $\omega=\mu^A$, i.e. if $\omega$ is the maximally mixed state on $A$,
\item[3.] $\p(\omega)=1$ if and only if $\omega$ is a pure state,
\item[4.] $\p(T\omega)=\p(\omega)$ for all reversible transformations $T\in\G_A$ and states $\omega\in\Omega_A$,
\item[5.] $\sqrt{\p}$ is convex, i.e.
\[
   \sqrt{\p\left(\sum_{i=1}^m \lambda_i \omega_i\right)}\leq \sum_{i=1}^m \lambda_i \sqrt{\p(\omega_i)}
\]
if $\lambda_i\geq 0$, $\sum_i \lambda_i=1$, and all $\omega_i\in\Omega_A$.
\end{itemize}
\end{lemma}
\proof
First, 5. follows directly from the fact that $\sqrt{\p(\omega)}=\|\omega\|$ is a norm (use the triangle inequality).
That pure states $\omega$ have $\p(\omega)=1$ follows directly from Definition~\ref{DefPurity}.
Since every state $\omega$ can be written as a convex
combination of pure states, it follows from 5. that $\p(\omega)\leq 1$ for all states $\omega$, and $\p(\omega)=\|\hat\omega\|^2\geq 0$
is clear. We have proven 1. Clearly, from $0=\p(\omega)=\|\hat\omega\|^2$, it follows that $\hat\omega=0$,
so $\omega=\mu^A$, and this proves 2. Since the inner product was chosen such that reversible transformations
are orthogonal, it follows that
\[
   \p(T\omega)=\langle T \hat\omega, T \hat \omega\rangle=\langle\hat \omega,\hat\omega\rangle=\p(\omega).
\]
Now consider the ball ${\cal B}:=\{x\in \hat A:\langle x,x\rangle\leq 1\}$. If $\omega$ is any state
with $\p(\omega)=1$, then $\hat \omega$ is on the surface of that ball; in particular, $\hat\omega$ is an exposed point of
the convex set ${\cal B}$. But since $\hat\omega\in \hat \Omega_A$ and $\hat\Omega_A\subset{\cal B}$, it must then also
be an exposed point of $\hat\Omega_A$, hence a pure state. This proves 3. Note that it also proves that all pure states are exposed.
\qed

In general, our definition of purity only works for transitive state spaces. Unfortunately, in the case of bipartite (and multipartite)
state spaces, this already excludes the most popular general probabilistic theory, colloquially called \emph{boxworld}. As it turns out, there
is a natural way to define an analogous notion of purity in boxworld, which we explain in Appendix~\ref{SecPurityBoxworld}.
However, the resulting notion of purity does not have all the nice properties of Lemma~\ref{LemPropPur} any more: in particular,
it equals unity for some pure (product) states, but is necessarily less than one for other pure (PR box) states.

In the quantum case, our definition of purity coincides with the standard definition up to a factor and an offset:
\begin{example}[Purity of Quantum States]
\label{ExPurityQuantum}
The real vector space which describes the states on a quantum $n$-level system is the set of Hermitian complex $n\times n$ matrices,
\[
   A=\{M\in\C^{n\times n}\,\,|\,\,M=M^\dagger\}.
\]
The cone of unnormalized states is given by all positive matrices, while the order unit is the trace functional:
\[
   A_+=\{M\in A\,\,|\,\, M\geq 0\},\qquad u^A(\rho)=\Tr(\rho).
\]
Thus, the set of normalized states $\Omega_A$ is the usual set of density matrices; similarly, the Bloch vector space $\hat A$
is the set of traceless Hermitian matrices. The group of reversible transformations $\G_A$
is the projective unitary group,
\[
   \G_A=\{U\cdot U^\dagger\,\,|\,\, U \in SU(n)\},
\]
and this group acts irreducibly on $\hat A$ (this follows from Lemma~\ref{LemClifford} in the appendix).
Thus, there is a unique inner product on $\hat A$ such that all reversible transformations are orthogonal. It is easy to guess
(we mentioned it before): it is the Hilbert-Schmidt inner product, scaled such that pure state Bloch vectors have norm 1:
\[
   \langle \hat L,\hat M\rangle:=\frac n {n-1}\Tr(\hat L \hat M)\qquad (\hat L,\hat M\in\hat A),
\]
As a consequence, the purity $\p(\rho)$ of any quantum state $\rho$ is
\begin{equation}
   \p(\rho)=\langle \hat \rho,\hat \rho\rangle=\langle \rho-\id/n,\rho-\id/n\rangle=\frac n {n-1}\Tr(\rho^2)-\frac 1 {n-1}.
   \label{eqQuantumPurity}
\end{equation}
\end{example}

Classical probability distributions can be treated in a similar manner:

\begin{example}[Purity of Classical Probability Distributions]
\label{ExPurityClassical}
The state space $\Omega_B$ of a classical $n$-level system is the set of all probability distributions on $n$ outcomes, that is, the simplex
\[
   \Omega_B=\left\{(p_1,\ldots,p_n)\,\,|\,\, p_i\geq 0, \sum_{i=1}^n p_i=1\right\}.
\]
This state space is contained in the vector space $B=\R^n$ with order unit $u^B(p):=\sum_{i=1}^n p_i$ for $p=(p_1,\ldots,p_n)\in B$.
The cone of unnormalized states is
\[
   B_+=\{p=(p_1,\ldots,p_n)\in B\,\,|\,\, p_i\geq 0 \mbox{ for all }i\},
\]
and the group of reversible transformation $\G_B$ is the permutation group $S_n$. The unique state which is invariant with respect to
$\G_B$ is the maximally mixed state $\mu^B=\left( \frac 1 n,\frac 1 n,\ldots,\frac 1 n\right)$.
It is a well-known fact of group representation theory~\cite{BarrySimon} that $\G_B$
acts irreducibly on $\hat B=\left\{p\in B\,\,|\,\, \sum_i p_i=0\right\}$. The unique invariant inner product on $\hat B$ turns out to be
\[
   \langle \hat p,\hat q\rangle:=\frac n {n-1}\sum_{i=1}^n \hat p_i \hat q_i\qquad (\hat p,\hat q\in \hat B).
\]
Permutations relabel the entries of $p$ and $q$ and preserve this inner product. Pure states $p=(0,\ldots,0,1,0,\ldots,0)$ have
$\langle \hat p,\hat p\rangle=1$. Thus, the purity of a probability distribution $p$ on $n$ outcomes, using $\hat p_i=p_i-1/n$, is
\begin{equation}
   \p(p)=\langle\hat p,\hat p\rangle=\frac n {n-1} \sum_{i=1}^n p_i^2 - \frac 1 {n-1}.
   \label{eqPurityClassical}
\end{equation}
This is the quantum result~(\ref{eqQuantumPurity}) restricted to diagonal matrices (as expected); however, here it is derived
without embedding the probability distributions into quantum state space.
\end{example}

\begin{example}[Purity for a gbit]
\label{ExPurityGbit}
Consider a generalized bit, or ``gbit''~\cite{Barrett07} where the state space $\hat\Omega_A$ is a square as in Figure~\ref{fig_gbit}. It can be understood as
describing ``one half of a PR box''~\cite{Barrett07}, and as a particular type of state space that appears in a theory called
\emph{generalized nonsignaling theory} or \emph{boxworld}~\cite{boxworld}.

We assume that the group of reversible transformations $\G_A$ is the group of all symmetries, which is consistent with
the tensor product in boxworld~\cite{boxworld}. Then $\G_A$ is the \emph{dihedral group} $D_4$, containing all
rotations of multiples of $\pi/2$ and reflections through diagonals. This group acts irreducibly on $\hat A=\R^2$. If we represent
$\hat\Omega_A$ as a square with $\hat\mu^A=0$ as the center, then the invariant inner product is given by the usual Euclidean
inner product. That is, the contour lines of constant purity correspond to circles in the state space, see Fig.\ref{fig_gbit}.
\end{example}

\subsection{Comparison to existing entropy measures}
\label{SubsecComparison}
The square state space (the ``gbit''), mentioned in Example~\ref{ExPurityGbit} and depicted in Figure~\ref{fig_gbit}, 
is also a good example to highlight a difference between the generalized purity used here and another possible generalization.

In the quantum case, the standard purity satisfies the equation $\Tr(\rho^2)=2^{-H_2(\rho)}$, where $H_2$ denotes the
R\'enyi entropy of order $2$. There has been some work on notions of entropy in general probabilistic theories~\cite{ShortWehner,Barnum,KimuraEntropy}. One could now imagine to define the purity of some state $\omega$ as $2^{-H_2(\omega)}$.

However, such a definition would have undesirable properties, as can be seen by example of the definitions in~\cite{ShortWehner}. Two possible definitions of $H_2$ are considered there: one possibility is to define the \emph{measurement entropy} $\hat H_2(\omega)$ of some state $\omega$ as the minimum R\'enyi-$2$ entropy of the set of outcome probabilities of any fine-grained measurement on $\omega$. However, using the measurement entropy, the corresponding definition of purity would assign purity $1$ to some highly mixed states in the boundary of the square state space, that is, the same value as for pure states. As an example, consider the fine-grained measurement on the gbit which consists of the two effects $E_1(\omega):=\sqrt{2}\hat\omega_y$ and $E_2(\omega):=1-\sqrt{2}\hat\omega_y$, if $\hat\omega=(\hat \omega_x,\hat\omega_y)$ denotes the Bloch vector corresponding to $\omega$ (that is, the corresponding point in the square). The state $\hat\omega=(0,1/\sqrt{2})$ is a mixed
state in the boundary of the square, but the measurement $(E_1,E_2)$ assigns outcome probabilities $0$ and $1$ to this state.
Hence $2^{-H_2(\omega)}=1$.
This shows that the measurement entropy can be misleading if used as a characterization of the mixedness of a state~\footnote{An interesting question is whether this or another suggested measure could quantify the free energy of a state, a property which has traditionally been used as a justification for choosing one entropy measure over another~\cite{VonNeumann55}.}. 

A second definition is the \emph{decomposition entropy}, $\breve H_2(\omega)$, which is defined as the minimum R\'enyi $2$-entropy of any probability distribution $(\lambda_1,\ldots,\lambda_n)$ with $n\in\N$ and $\sum_i \lambda_i=1$ such that
$\omega=\sum_i \lambda_i \omega_i$, with pure states $\omega_i$. That is, it is the minimal entropy of the coefficients in any
decomposition of $\omega$ into pure states. However, as shown in~\cite[(D8)]{ShortWehner}, there are states $\omega_1,\omega_2$ in the gbit state space with the property that $\breve H_2\left(\frac 1 2 \omega_1+\frac 1 2 \omega_2\right)<\frac 1 2 \breve H_2(\omega_1)+\frac 1 2 \breve H_2(\omega_2)$. According to the corresponding purity definition, the mixture $\frac 1 2 \omega_1+\frac 1 2 \omega_2$ would have higher purity then both $\omega_1$ and $\omega_2$, violating intuition about mixtures being ``at least as mixed'' as their components. In contrast, it follows from property 5.\ in Lemma~\ref{LemPropPur} that our notion of purity $\p$ always satisfies $\p\left(\frac 1 2 \omega_1+\frac 1 2 \omega_2\right)\leq\max\{ \p(\omega_1),\p(\omega_2)\}$.

Another advantage of our definition of purity, as compared to other approaches, is that it satisfies several useful identities arising
from group theory. Thus, it is sometimes possible to calculate its value explicitly on the basis of
simple properties of the state space (such as in Theorem~\ref{TheCompositePurity}), which is important to derive the results of
this paper. This is analogous to the situation in quantum theory, where purity is often used as an easy-to-calculate replacement for von
Neumann entropy.

 \begin{figure}[!hbt]
 \begin{center}
 \includegraphics[angle=0, width=4cm]{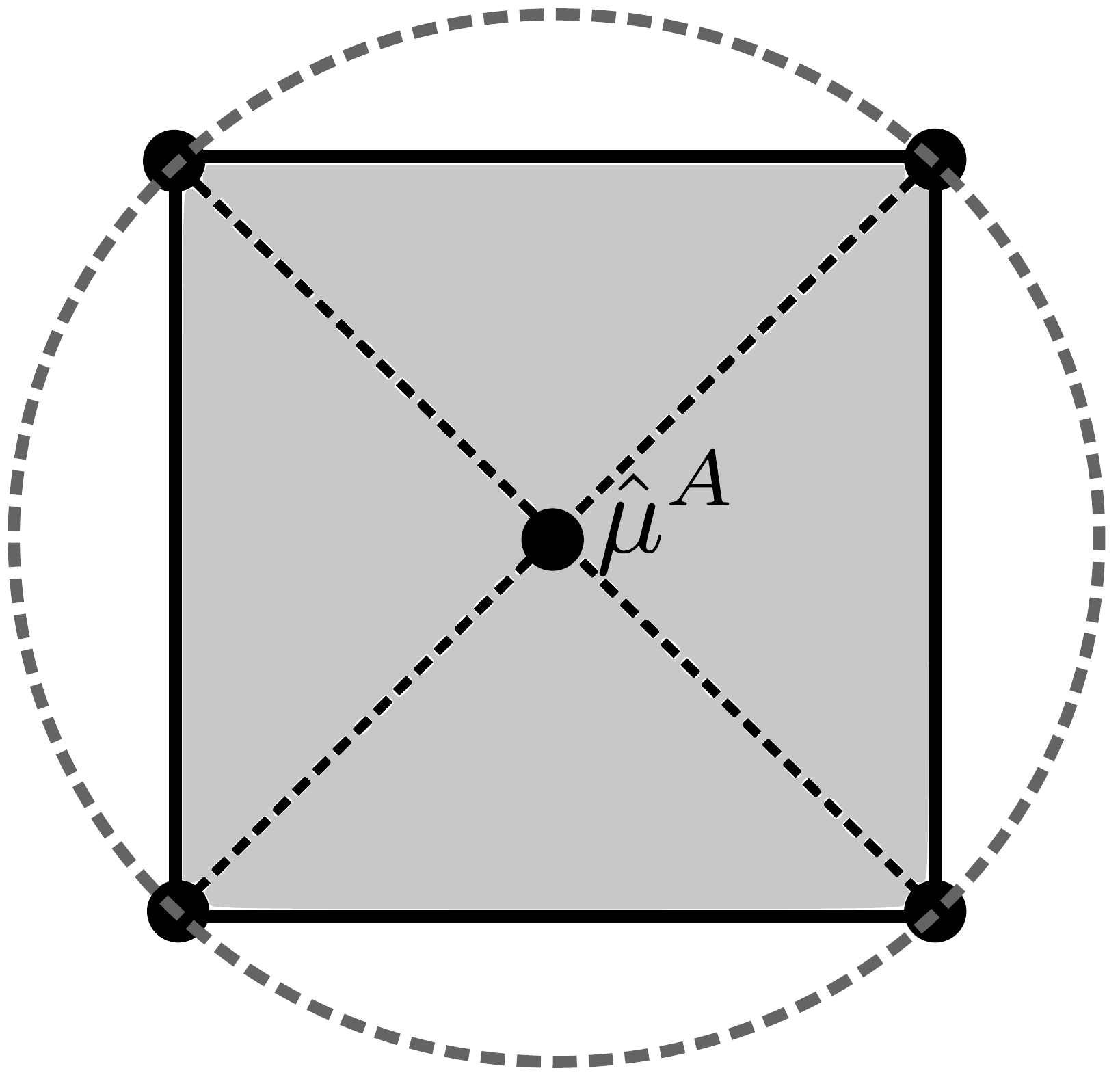}\hskip 2cm
 \includegraphics[angle=0, width=4.3cm]{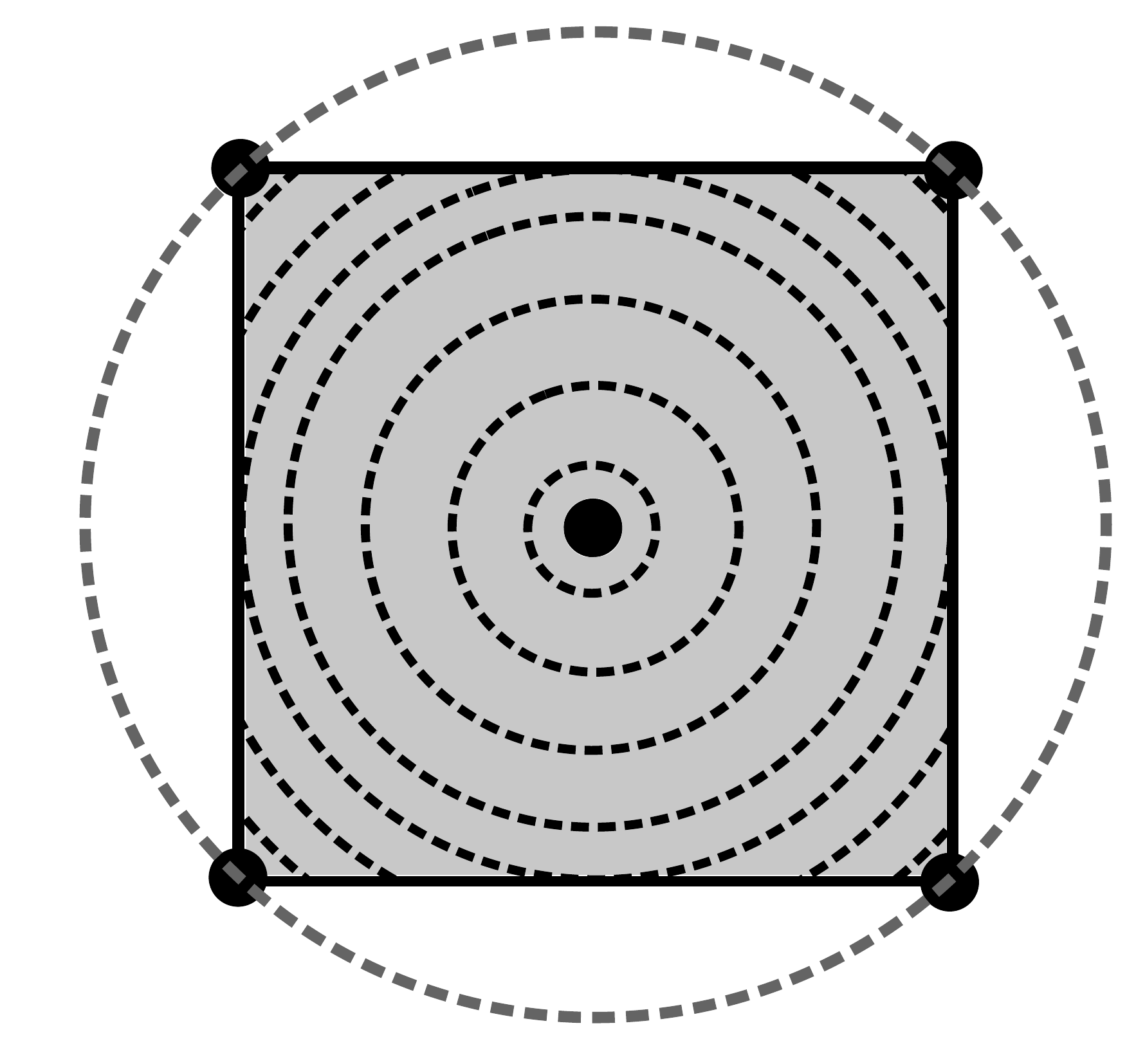}\hskip 2cm
 \includegraphics[angle=0, width=4cm]{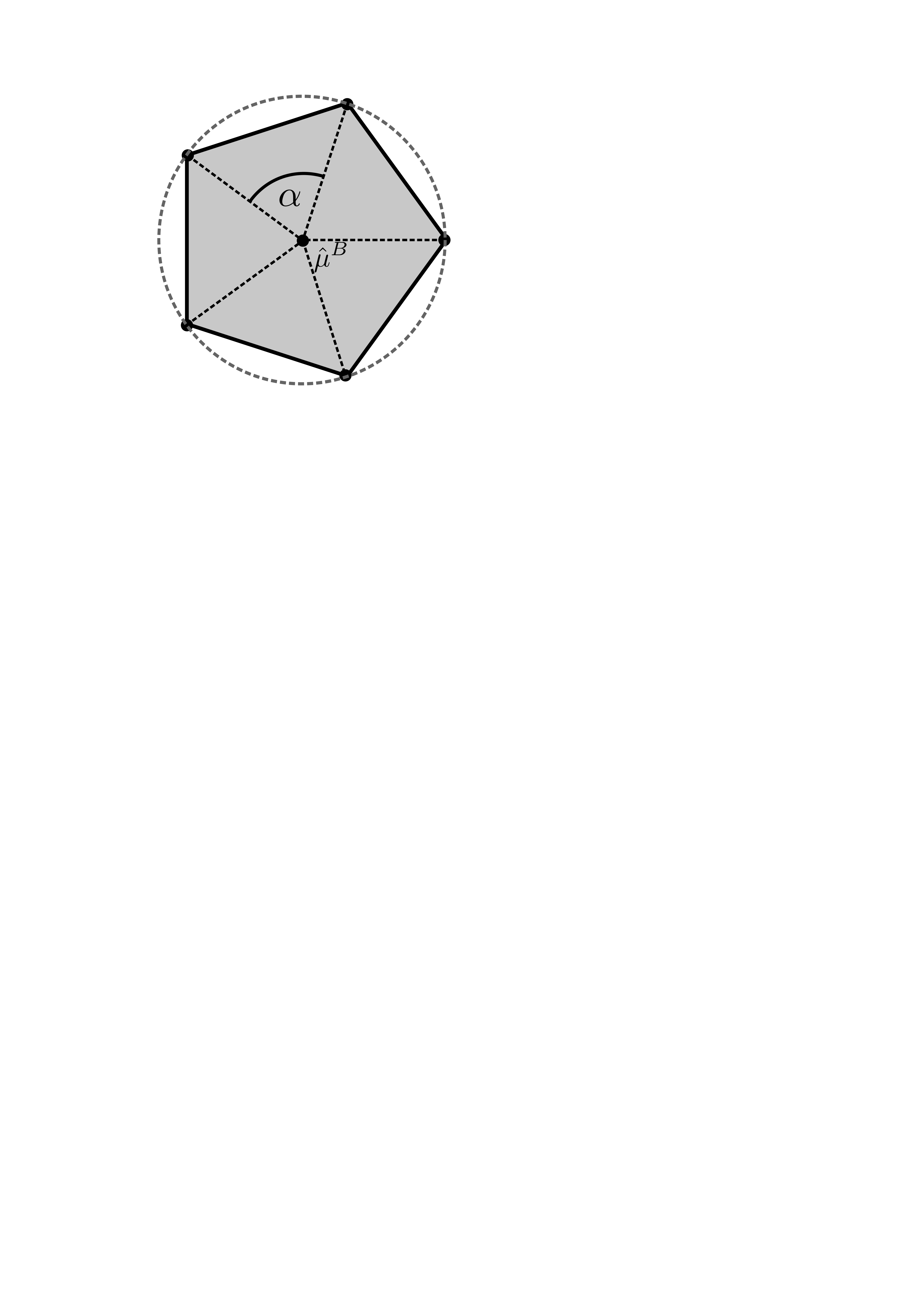}
 \caption{The left and center panels display a ``gbit'' state space $A$, where $\hat\Omega_A$ is a square and the capacity
 is $N_A=2$; shown are the four pure states and the maximally mixed state $\hat\mu^A=0$. The symmetry group is the dihedral group $\G_A=D_4$, and the contour lines of constant purity are circles (cf.\ Example~\ref{ExPurityGbit}). It has a complete set of Paulis, consisting of $2$ maps which is the minimal possible number. The right panel shows a pentagon state
 space $\hat\Omega_B$. The group of symmetries is the dihedral group $\hat\G_B=D_5$, containing rotations of multiples
 of $\alpha=2\pi/5$. This state space has capacity $N_B=2$, a complete set of Paulis consists of $5$ maps, and the maximally
 mixed state cannot be written as a uniform mixture of two perfectly distinguishable pure states, in contrast to the square.}
 \label{fig_gbit}
 \end{center}
 \end{figure}

\subsection{Irreducible subgroups and generalized Paulis}
\label{SubsecPaulis}
In the quantum case, there is a simple formula expressing purity in terms of squared expectation values of Pauli operators.
For a single qubit, denote the $2\times 2$ Pauli matrices by $(X_0,X_1,X_2,X_3):=(\id,X,Y,Z)$, then
$\displaystyle \Tr(\rho^2)=\frac 1 2 \sum_{i=0}^3 \left(\Tr( X_i\rho)\right)^2$. A similar formula holds in the case of
several qubits; we discuss this below.

As it turns out, there is an interesting generalization of these identities to the general probabilistic case. To understand the general case,
it is useful to think of the Pauli operators not as matrices, but as linear maps that assign real numbers to states:
$\rho\mapsto \Tr(X_i\rho)$. These maps have certain properties that correspond to the conditions in the following definition:
\begin{definition}
\label{DefPauli}
Let $A$ be a transitive state space. A linear map $X:A\to\R$ is called a \emph{Pauli map} if
\begin{itemize}
\item $X(\mu^A)=0$ for the maximally mixed state $\mu^A\in\Omega_A$, and
\item $\max\{|X(\hat a)|\,\,|\,\, \hat a \in \hat A,\enspace \langle\hat a,\hat a\rangle\leq 1\}=1$.
\end{itemize}
\end{definition}
If $X$ is a Pauli map, then there exists a vector $\hat X\in \hat A$ such that $X(\hat a)=\langle\hat X,\hat a\rangle$ for all $\hat a \in \hat A$.
Thus, the second condition in Definition~\ref{DefPauli} is equivalent to the condition $\|\hat X\|_2=1$, where
$\|\hat X\|_2=\sqrt{\langle\hat X,\hat X\rangle}$ denotes the norm on $\hat A$ which is derived from the invariant inner product
$\langle\cdot,\cdot\rangle$ on $\hat A$.

In the quantum case of a single qubit, it is easy to see that the maps $\rho\mapsto \Tr(X_i \rho)$ are Pauli maps if $i\in\{1,2,3\}$, but not if $i=0$:
recall the definitions in Example~\ref{ExPurityQuantum}, and let $i\in\{1,2,3\}$, then
\[
   \rho\mapsto \Tr(X_i\rho)=\Tr(X_i\hat\rho)\stackrel ! = \langle\hat X_i,\hat\rho\rangle=2\Tr(\hat X_i\hat\rho),
\]
which proves that this map is represented by the vector (traceless Hermitian matrix) $\hat X_i=\frac 1 2 X_i\in\hat A$, that is,
one half times the corresponding Pauli matrix. Then, we have for the norm on $\hat A$
\[
   \|\hat X_i\|_2^2=\langle\hat X_i,\hat X_i\rangle=2\Tr(\hat X_i^2)=2\Tr\left(\frac 1 4 \cdot\id\right)=1
\]
which proves that the corresponding maps are Pauli maps in the sense of Definition~\ref{DefPauli}.

Pauli maps are related to purity by the following lemma:
\begin{lemma}
\label{LemFormulaPauliPurity}
Let $A$ be an irreducible transitive dynamical state space of dimension $K_A$, and let $\mathcal{H}\subseteq \G_A$
be a compact subgroup which acts irreducibly on $\hat A$ (for example, $\mathcal{H}=\G_A$). Then, if $X$ is any Pauli map on $A$,
\[
   \int_{H\in\mathcal{H}}\left(X\circ H(\omega)\right)^2\, dH=\frac{\p(\omega)}{K_A-1}\qquad\mbox{ for all states }\omega\in\Omega_A.
\]
\end{lemma}
\proof
If $M\geq 0$ is any positive matrix on $\hat A$, then $I:=\int_{H\in\mathcal{H}} H M H^{-1}\, dH\geq 0$ satisfies
$[I, H]=0$ for all $H\in\mathcal{H}$. Thus, by Schur's Lemma, we have $I=c\cdot\Id_{\hat A}$ for some $c\in\R$.
By taking the trace of both sides, we see that $c=\Tr M/(K_A-1)$.
Since $\left(X\circ H(\omega)\right)^2 = \langle\hat X | H |\hat\omega\rangle\langle\hat\omega |H^{-1}|\hat X\rangle$, we get
\begin{eqnarray*}
   \int_{H\in\mathcal{H}} \left(X\circ H(\omega)\right)^2\, dH&=&
   \langle\hat X|\left( \int_{H\in\mathcal{H}}  H |\hat\omega\rangle\langle\hat\omega| H^{-1}\, dH\right)|\hat X\rangle
   =\langle\hat X|\frac{\Tr |\hat\omega\rangle\langle\hat\omega|}{K_A-1}\Id_{\hat A}|\hat X\rangle
   =\frac{\langle\hat\omega|\hat\omega\rangle\langle\hat X|\hat X\rangle}{K_A-1}.
\end{eqnarray*}
Note that in the quantum case, the vectors themselves are \emph{operators}, hence the matrices and linear maps
appearing in this calculation are \emph{superoperators}.
\qed

The result becomes particularly interesting if the subgroup $\mathcal{H}\subseteq\G_A$ is finite: it will finally give the natural
analog of Pauli matrices in more general theories.
\begin{corollary}
\label{CorPauli}
Let $A$ be an irreducible transitive dynamical state space of dimension $K_A$, and let $\mathcal{H}\subseteq \G_A$
be a \emph{finite} subgroup which acts irreducibly on $\hat A$. Fix any Pauli map $X_1$ on $A$, and let $\mathcal{X}$ be
the orbit of $\mathcal{H}$ on $X_1$, disregarding the sign of each map. That is, $\mathcal{X}:=\{X_1\circ H\,\,|\,\, H\in\mathcal{H}\}/\{+1,-1\}$.
Then
\[
   \frac 1 {|\mathcal{X}|} \sum_{X\in\mathcal{X}}\left(X(\omega)\right)^2 = \frac{\p(\omega)}{K_A-1},
\]
and we call $\mathcal{X}$ a \emph{complete set of Paulis} for $A$. Note that $X(\omega)=\langle\hat X,\hat \omega\rangle$.
\end{corollary}
For a classical probability distribution $p=(p_1,\ldots,p_n)$, the expression $\sum_{i=1}^n p_i^2$
is sometimes called the \emph{collision probability}. It is directly related to our notion of purity $\p(p)$ by eq.~(\ref{eqPurityClassical}),
and can be interpreted as the probability that two identically prepared copies of $p$ give the same outcome, if the random
variable $i\in\{1,\ldots,n\}$ is measured. A similar interpretation exists in the quantum case: we can ask for the probability of getting the same outcome,
if we measure two copies of a state $\rho$ in a fixed basis. Maximizing this collision probability over all bases yields $\Tr(\rho^2)$,
with the maximum being attained in the eigenbasis.

The following lemma generalizes this observation to other probabilistic theories.
\begin{lemma}[Operational Interpretation of Purity]
Any Pauli map $X$ can be interpreted as a measurement, giving outcomes $\pm 1$ on
a state $\omega$ with probabilities $(1\pm X(\omega))/2$. The corresponding expectation value is exactly $X(\omega)$.

Denote by $\mathbb{P}_\omega^c(X)$ the probability that two successive measurements of $X$ on two identically prepared
copies of $\omega$ give the same outcome (``$c$'' is for ``collision probability''). Then it turns out that
\[
   \max_{X\mbox{ Pauli map}} \mathbb{P}_\omega^c(X)=\frac 1 2 \left(\strut 1+\p(\omega)\right).
\]
\end{lemma}
\proof
We use the Cauchy-Schwarz inequality
\[
   \mathbb{P}_\omega^c(X)=\left(\frac{1+X(\omega)} 2 \right)^2 +\left(\frac{1-X(\omega)} 2 \right)^2 =\frac 1 2 \left(1+X(\omega)^2\right)
   =\frac 1 2 +\frac 1 2 \langle\hat X,\hat\omega\rangle^2\leq \frac 1 2 +\frac 1 2 \|\hat X\|_2^2\cdot \|\hat\omega\|_2^2
   =\frac 1 2 \left(\strut 1+\p(\omega)\right).
\]
This upper-bound is attained on the Pauli map corresponding to $\hat X:=\hat \omega/\|\hat\omega\|_2$.
\qed

In quantum theory on $k$ qubits, our notion of a ``complete set of Paulis'' reduces to the usual Pauli operators:
\begin{example}[Paulis on $k$ Qubits in Quantum Theory]
\label{ExQuantumPauli}
Recall the quantum situation described in Example~\ref{ExPurityQuantum}, but now on $n=2^k$-dimensional Hilbert space,
i.e.\ $A$ is the quantum state space of $k$ qubits. A particular finite subgroup of $\G_A$ is given by the \emph{Clifford group}~\cite{Gottesman98}
\[
   C_k:=\{U\in U(2^k)\,\,|\,\, U P U^\dagger\in P_k\mbox{ for all }P\in P_k\},
\]
where $P_k$ is the \emph{Pauli group} on $k$ qubits, i.e.\ $P_k=\{\pm\sigma_1\otimes\ldots\otimes\sigma_k\,\,|\,\,\sigma_i
\in\{\id,X,Y,Z\}\}$. This group acts irreducibly by conjugation on $\hat A$, the set of traceless Hermitian matrices on $(\C^2)^{\otimes k}$;
we show this in Lemma~\ref{LemCliffordIrrep} in Appendix~\ref{AppIrr}. Consider $X^{\otimes k}$, the $k$-fold tensor product of the Pauli matrix $X$.
We would like to find a constant $c>0$ such that $X_1(\rho):=c\cdot\Tr(X^{\otimes k}\rho)$ becomes a Pauli map.
First, we calculate the vector (i.e.\ matrix) $\hat X_1$ which describes $X_1$:
\[
   X_1(\rho)=c\cdot\Tr(X^{\otimes k}\rho)=c\cdot\Tr(X^{\otimes k}\hat \rho)\stackrel ! = \langle\hat X_1,\hat \rho\rangle
   =\frac{2^k}{2^k-1} \Tr(\hat X_1\hat \rho),
\]
and we see that $\hat X_1=c(2^k-1)2^{-k}X^{\otimes k}$. The constant $c$ is determined by normalization:
\[
   1\stackrel ! = \langle\hat X_1,\hat X_1\rangle=\frac{2^k}{2^k-1}\Tr\left(\hat X_1^2\right)=\frac{2^k}{2^k-1}\cdot
   \frac{c^2(2^k-1)^2}{2^{2k}} \Tr\left((X^{\otimes k})^2\right)=c^2(2^k-1),
\]
hence $c=1/\sqrt{2^k-1}$. Now if $H=U\cdot U^\dagger$ with $U\in C_k$, then
\[
   (X_1\circ H)(\rho)=\frac 1 {\sqrt{2^k-1}}\cdot \Tr(X^{\otimes k}H(\rho))=\frac 1 {\sqrt{2^k-1}}\cdot \Tr(X^{\otimes k}U\rho U^\dagger)
   =\frac 1 {\sqrt{2^k-1}}\cdot \Tr(U^\dagger X^{\otimes k} U \rho)
\]
By choosing appropriate elements $H=U^\dagger\cdot U$ of the Clifford group, the matrix $X^{\otimes k}\in P_n$ is mapped
to every other element of $P_n$ except the identity. Ignoring the sign as suggested in Corollary~\ref{CorPauli},
we get the orbit
\[
   \mathcal{X}=\left\{ \rho\mapsto \frac 1 {\sqrt{2^k-1}}\Tr(\sigma_1\otimes\ldots\otimes\sigma_k\,\rho)\,\,\left|
   \,\, \sigma_i\in\{\id,X,Y,Z\},\mbox{ not all }\sigma_i=\id\right.\right\}.
\]
This is a ``complete set of Paulis'' according to Corollary~\ref{CorPauli}: these are the (maps corresponding to)
the usual Pauli matrices.
Therefore, purity can be expressed as
\[
   \frac 1 {4^k-1}\sum_{(\sigma_1,\ldots,\sigma_k)\neq(\id,\ldots,\id)} \frac{\left(\Tr(\sigma_1\otimes\ldots\otimes\sigma_k \rho)\right)^2}
   {2^k-1}=\frac{\p(\rho)}{4^k-1}.
\]
The standard result $\displaystyle \Tr(\rho^2)=2^{-k}\sum_{\sigma_1,\ldots,\sigma_k} \Tr(\sigma_1\otimes
\ldots\otimes\sigma_k\rho)^2$ follows from some further simplification.
\end{example}

For a classical $n$-level system introduced in Example~\ref{ExPurityClassical}, a complete set of Paulis consists of $n$ maps that
basically read out a probability vector's components:
\begin{example}[Paulis in Classical Probability Theory]
With the notation of Example~\ref{ExPurityClassical}, let $X_1:B\to\R$ be the map
\[
   X_1(p):=p_1 - \frac 1 {n-1}\sum_{i=2}^n p_i.
\]
For the maximally mixed state $\mu^B=\left( \frac 1 n,\ldots,\frac 1 n \right)$, we have $X_1(\mu^B)=0$. It is easy to check
that $X_1(\hat p)=\langle \hat X_1,\hat p\rangle$ for the invariant inner product on $\hat B$ if
$\displaystyle \hat X_1=\left(\frac{n-1} n , -\frac 1 n,\ldots,-\frac 1 n\right)$. Moreover, we have $\langle\hat X_1,\hat X_1\rangle=1$,
hence $X_1$ is a Pauli map according to Definition~\ref{DefPauli}. Now let $\mathcal{H}=\G_B=S_n$ be the full permutation group.
If $\sigma\in\mathcal{H}$ is any permutation with, say, $\sigma(i)=1$, then
$X_i(p):=\displaystyle X_1\circ \sigma(p)=p_i - \frac 1 {n-1}\sum_{j\neq i} p_j$ (for normalized probability vectors $p\in\Omega_B$,
this is just $X_i(p)=\frac n {n-1} p_i - \frac 1 {n-1}$).
Thus, a complete set of Paulis $\mathcal{X}$ is given by the set of maps
\[
   \mathcal{X}=\left\{\left.
      p\mapsto p_i - \frac 1 {n-1}\sum_{j\neq i}p_j\,\,\right|\,\,
      i\in\{1,\ldots,n\}
   \right\}.
\]
Then the formula from Corollary~\ref{CorPauli} reproduces eq.~(\ref{eqPurityClassical}).
\end{example}
\begin{example}[Paulis for Polygonal State Spaces]
Consider state spaces which are regular polygons; that is, $\hat\Omega$ is a regular $n$-gon inscribed in the unit circle
as in Figure~\ref{fig_gbit}. Then complete sets of Paulis (in the sense of Corollary~\ref{CorPauli}) look very differently, depending
on properties of the symmetry group $D_n$, the dihedral group. We illustrate this for the cases $n=4$ and $n=5$.

Let $\hat X_1:=(1,0)$, such that the corresponding Pauli map acts on the Bloch space $\R^2$ via $X_1(\hat \omega)=
\langle \hat X_1,\hat\omega\rangle=\omega_1$,
if $\hat \omega=(\omega_1,\omega_2)$. First, consider a ``gbit'' system $A$ where the Bloch representation of state space, $\hat \Omega_A$,
is a square, inscribed into a unit circle as in Figure~\ref{fig_gbit}.
The symmetry group is the dihedral group $D_4$; its orbit on $X_1$ consists of the maps
\[
   \{X_1\circ G\,\,|\,\, G\in D_4\}=\{(\omega_1,\omega_2)\mapsto \omega_1, (\omega_1,\omega_2)\mapsto -\omega_1,
   (\omega_1,\omega_2)\mapsto \omega_2, (\omega_1,\omega_2)\mapsto -\omega_2
   \}.
\]
Disregarding the sign, we get a complete set of Paulis in the sense of Corollary~\ref{CorPauli}, which is $\mathcal{X}=\{ X_1, X_2\}$,
where $X_1(\hat\omega)=\omega_1$ and $X_2(\omega)=\omega_2$. Since $K_A=3$, the formula from Corollary~\ref{CorPauli} becomes
\begin{equation}
   \frac 1 2 (\omega_1^2+\omega_2^2)=\frac{\p(\omega)} 2 \mbox{ for all }\omega\in\Omega_A,
   \label{eqGeometricPurity}
\end{equation}
which just expresses the fact that the purity equals the squared Euclidean length of the Bloch vector.

Now consider the case $n=5$, that is, a state space $B$ where $\hat\Omega_B$ is a regular pentagon. A smallest irreducible
subgroup is given by $\mathcal{H}:=\left\{R(2 k\pi/5)\,\,|\,\, k=0,1,\ldots,4\right\}$ (denoting its action on Bloch vectors), where $R(\alpha)$
denotes rotation by angle $\alpha$ in $\R^2$. In general, $(X_1\circ R(2 k\pi/5)(\hat \omega))^2$ gives different values for
all $k$, which means that a complete set of Paulis necessarily contains all the five maps. Thus, all we get is
\begin{equation}
   \frac 1 5 \sum_{k=0}^4 \left(\strut X_1(R(2k\pi/5)\hat \omega\right)^2=\frac{\p(\omega)}2 \mbox{ for all }\omega\in\Omega_B.
   \label{eqFive}
\end{equation}
In a sense, this is ``inefficient'': in order to compute $\p(\omega)=\|\hat\omega\|^2$, we could as well use eq.~(\ref{eqGeometricPurity}),
which involves only two addends instead of five. However, the advantage of~(\ref{eqFive}) compared to~(\ref{eqGeometricPurity}) is
that all involved maps $X_k:=X_1\circ R(2k\pi/5)$ are ``equivalent'' for the state space $B$: they are all connected by reversible
transformations. In other words, in order to build a device that measures $X_k$, it is sufficient to have a device measuring $X_1$.
All other measurements can then be accomplished by composing $X_1$ with a reversible transformation $T_k:=R(2k\pi/5)$,
as sketched in Figure~\ref{fig_detector}. Within state space $B$, this is not possible for the two maps $X_1(\omega):=\omega_1$ and
$X_2(\omega):=\omega_2$ that appear in eq.~(\ref{eqGeometricPurity}).
\end{example}
 \begin{figure}[!hbt]
 \begin{center}
 \includegraphics[angle=0, width=8cm]{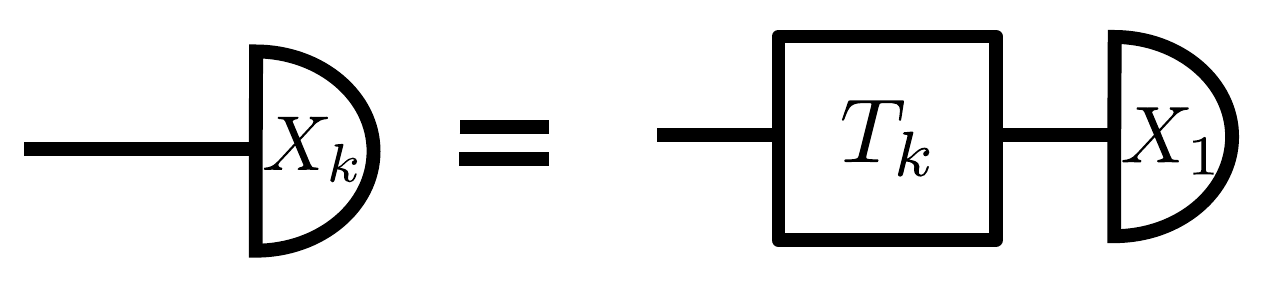}
 \caption{All elements $X_k$ of a complete set of Paulis (in the sense of Corollary~\ref{CorPauli}) on a general state space are connected
 by reversible transformations. That is, every $X_k$ can be measured
 by first applying a reversible transformation $T_k$, and then measuring a fixed Pauli $X_1$.}
 \label{fig_detector}
 \end{center}
 \end{figure}

\subsection{Bipartite systems: local purity as an entanglement measure}
\label{SubsecBipartite}
The main goal of this paper is to investigate typical states on composite state spaces. When we have state spaces $A$
and $B$, there are in general many different possible ways to combine them into a joint state space $AB$.
However, there is a minimal set of assumptions that necessarily must hold in order to interpret $AB$ as a ``joint state
space'' in a physically meaningful way. The most important assumption is \emph{no-signalling}: measurements
on one subsystem do not affect the outcome probabilities on other subsystems. In this section, we make
an additional simplifying assumption which is often (but not always) imposed in the framework of general probabilistic theories:
that of \emph{local tomography}. However, we will later drop this assumption in Subsection~\ref{SubsecNotLocallyTomo}.

{\bf Assumption: Local tomography.} If $A$ and $B$ are state spaces, then the joint state space $AB$ has the property
that \emph{states $\omega^{AB}\in\Omega_{AB}$ are uniquely characterized by the outcome probabilities of the
local measurements on $A$ and $B$.}

From a physics point of view, this assumption means that the content of bipartite states consists of the \emph{correlations}
of outcome probabilities of local measurements. This is equivalent to the multiplicativity of the state space dimension:
$K_{AB}=K_A K_B$. It can be shown~\cite{Barrett07} that this assumption implies the tensor product formalism:
The linear space which carries the global unnormalized states is the algebraic tensor product of the local spaces:
$AB=A\otimes B$. We have the notion of product states $\omega^A\otimes\omega^B\in\Omega_{AB}$ for
states $\omega^A\in\Omega_A$, $\omega^B\in\Omega_B$, and similarly for effects, with the same interpretation as
in the quantum case. In particular, the unit effect on $AB$ is $u^{AB}=u^A\otimes u^B$. For global states $\omega^{AB}\in\Omega_{AB}$,
we can define the \emph{reduced state} $\omega^A\in\Omega_A$ by $L^A(\omega^A):=L^A\otimes u^B(\omega^{AB})$ for all
linear functionals $L^A$ on $A$ (in particular, for all effects).

In accordance with~\cite{Barrett07}, we give a list of additional assumptions that naturally follow from the physical interpretation
of a composite state space. First, if $\omega^A\in\Omega_A$ and $\omega^B\in\Omega_B$, then we assume that
$\omega^A\otimes\omega^B\in\Omega_{AB}$. That is, we assume that it is possible to prepare states independently
on $A$ and $B$. Second, if $\omega^{AB}\in\Omega_{AB}$ is any global state, we assume that the local reduced states
are valid states on $A$ and $B$: $\omega^A\in \Omega_A$ and $\omega^B\in\Omega_B$. Since this work is on
\emph{dynamical} state spaces, we also postulate that reversible transformations can always be applied locally. That is,
$\G_A\otimes\G_B\subseteq \G_{AB}$.

Similarly as in the quantum case, a global state $\omega^{AB}\in\Omega_{AB}$ will be called \emph{entangled} if it
cannot be written as a convex combination of product states. Now suppose $\omega^{AB}$ is pure.
\begin{itemize}
\item If the purity of the local reduced state is one, i.e.\ $\p(\omega^A)=1$, then $\omega^A$ must be pure. From this,
it follows~\cite{HardyFoliable} that $\omega^{AB}=\omega^A\otimes\omega^B$ -- that is, the global state is unentangled.
\item On the other hand, if $\p(\omega^A)<1$, then $\omega^{AB}$ cannot be written as a product $\omega^A\otimes\omega^B$,
since $\omega^A$ would necessarily have to be pure. Thus, $\omega^{AB}$ is entangled.
\end{itemize}
That is, the local purity $\p(\omega^A)$ can be understood as an entanglement measure:
the smaller $\p(\omega^A)$, the ``more entangled'' $\omega^{AB}$. If $\p(\omega^A)=0$, or equivalently
$\omega^A=\mu^A$, we may call $\omega^{AB}$ \emph{maximally entangled}. It turns out that a PR box is
an example of a maximally entangled post-quantum state in this sense~\cite{boxworld}.

It is natural to ask for the \emph{typical entanglement} of random pure states on a composite state space $AB$.
As discussed in the context of Definition~\ref{DefTransitivity} above, in order for this notion to make sense, we
need the property of \emph{transitivity}: for every pair of pure states $\alpha,\omega\in\Omega_{AB}$,
there must be a reversible transformation $T\in\G_{AB}$ such that $T\alpha=\omega$. It is important to note that
transitivity of the local state spaces $A$ and $B$ does \emph{not} imply transitivity of the joint state space $AB$.
A simple example is given by a state space called ``boxworld''~\cite{boxworld}: suppose that $A$ and $B$
are both square state spaces as in the left of Figure~\ref{fig_gbit}, and $AB$ is the state space which contains all
no-signalling behaviours (including, for example, PR-box states).
That is, $\Omega_{AB}$ is assumed to be the largest possible subset of $AB$ that is
consistent with the assumptions mentioned above ($\Omega_{AB}$ is sometimes called the ``no-signalling polytope'',
or the ``maximal tensor product'' of the local state spaces). Then it turns out that the global state space is not transitive:
for example, no reversible transformation takes a pure product state to a pure PR-box state.

Thus, in the following, we will only consider composite state spaces $AB$ that are themselves transitive. As a first observation,
it turns out that the maximally mixed state on $AB$ is the product of the maximally mixed states of $A$ and $B$.
The proof is given in~\cite{MasanesMueller}.
\begin{lemma}
If $A$, $B$, and $AB$ are transitive dynamical state spaces, then $\mu^{AB}=\mu^A\otimes \mu^B$.
\end{lemma}
If $AB$ is transitive, we can decompose it into the Bloch subspace and multiples of the maximally mixed state:
$AB=(AB)^\wedge\oplus \R\mu^{AB}$. On the other hand, if $A$ and $B$ are transitive, we can substitute their
local decompositions into the tensor product:
\[
   AB\equiv A\otimes B=(\hat A\oplus\R\mu^A)\otimes(\hat B\oplus\R\mu^B)
   =(\hat A\otimes \hat B)\oplus(\hat A\otimes\mu^B)\oplus(\mu^A\otimes\hat B)\oplus\R\mu^{AB}.
\]
Since $u^{AB}=u^A\otimes u^B$, the unit effect is zero on the first three addends in this decomposition.
This shows that
\begin{equation}
   (AB)^\wedge=(\hat A\otimes \hat B)\oplus(\hat A\otimes\mu^B)\oplus(\mu^A\otimes\hat B).
   \label{eqDecomp}
\end{equation}
This decomposition is reminiscent of another ``Bloch representation''~\cite{DakicBrukner,MasanesMueller} which writes global states in terms of three
vectors: the two local reduced states, and a correlation matrix. Now suppose that, in addition, $AB$ is irreducible.
Then there is a unique inner product on $(AB)^\wedge$ such that all transformations $T\in\G_{AB}$ are orthogonal.
Moreover, the three subspaces in eq.~(\ref{eqDecomp}) are preserved by local transformations $T_A\otimes T_B$, and they
are mutually orthogonal. To see this, first note that for pure states $\psi\in\Omega_A$, we have $\int_{G_A\in\G_A} G_A \hat\psi\, dG_A
=\int_{G_A\in\G_A} G_A\psi\, dG_A - \mu^A=0$. Since the pure states span $A$, the $\hat\psi$ span $\hat A$, and this integral
is zero for all vectors $\hat\psi\in\hat A$. Now mutual orthogonality of the subspaces, for example $\hat A \otimes\hat B \perp \hat A \otimes \mu^B$,
follows from
\begin{eqnarray*}
   \langle\hat a\otimes\hat b,\hat a'\otimes\mu^B\rangle&=&\langle(\Id_A\otimes G_B)(\hat a\otimes\hat b),(\Id_A\otimes G_B)
   (\hat a'\otimes \mu^B)\rangle=\langle\hat a\otimes G_B \hat b,\hat a'\otimes\mu^B\rangle
   =\int_{G_B\in\G_B} \langle\hat a\otimes G_B \hat b,\hat a'\otimes\mu^B\rangle\, dG_B \\
   &=& \langle \hat a\otimes 0,\hat a'\otimes\mu^B\rangle=0
\end{eqnarray*}
(the other pairs of subspaces can be treated similarly). The value of the inner product on $\hat A\otimes\mu^B$ (and similarly
on $\mu^A\otimes\hat B$) can be calculated explicitly:
\begin{lemma}
\label{LemInnerProdMaxMix}
Let $A$, $B$, and $AB$ be transitive dynamical state spaces, where $A$ and $AB$ are irreducible. Then
\[
   \langle \hat x\otimes\mu^B,\hat y\otimes\mu^B\rangle=\p(\varphi^A\otimes\mu^B)\langle\hat x,\hat y\rangle
   \qquad\mbox{for all }\hat x,\hat y\in\hat A,
\]
where $\varphi^A$ is any pure state on $A$.
\end{lemma}
\proof
For any pair of vectors $\hat x,\hat y\in\hat A$, define $(\hat x,\hat y):=\langle \hat x\otimes\mu^B,\hat y\otimes \mu^B\rangle$.
Clearly, this is an inner product on $\hat A$. Moreover, it is invariant with respect to all reversible transformations on $A$. Explicitly,
for all $T\in\G_A$,
\[
   (T\hat x,T \hat y)=\langle T\hat x \otimes\mu^B,T\hat y\otimes\mu^B\rangle = \langle (T\otimes\Id)(\hat x\otimes\mu^B),
   (T\otimes\Id)(\hat y \otimes\mu^B)\rangle=\langle \hat x\otimes \mu^B,\hat y\otimes \mu^B\rangle=(\hat x,\hat y),
\]
since local transformations are in particular orthogonal with respect to the global invariant inner product. According to
Lemma~\ref{LemInvariantInnerProduct}, it follows that there exists some global constant $c>0$ such that
$(\hat x,\hat y)=c\langle\hat x,\hat y\rangle$. Choosing $\hat x=\hat y=\hat\varphi^A$ for some pure state $\varphi^A\in\Omega_A$
shows that $c=(\hat\varphi^A,\hat\varphi^A)=\langle\hat\varphi^A\otimes\mu^B,\hat\varphi^A\otimes\mu^B\rangle$. But
\begin{equation}
   (\varphi^A\otimes\mu^B)^\wedge = (\hat\varphi^A+\mu^A)\otimes \mu^B-\mu^A\otimes\mu^B=\hat\varphi^A\otimes\mu^B,
   \label{eqBlochMaxMix}
\end{equation}
hence $c=\langle(\varphi^A\otimes\mu^B)^\wedge,(\varphi^A\otimes\mu^B)^\wedge\rangle=\p(\varphi^A\otimes\mu^B)$.
\qed

We would like to construct Pauli maps on $AB$ from Pauli maps on $A$ and $B$. In particular, if $X^A$ is a Pauli map
on $A$, a natural idea is to use the map $X^A\otimes u^B$ on the global state space. Is this a Pauli map? First, we
have $X^A\otimes u^B(\mu^{AB})=X^A(\mu^A)u^B(\mu^B)=0$, so the first condition of Definition~\ref{DefPauli} is satisfied.
But there is a second condition, demanding that the vector $(X^A\otimes u^B)^\wedge$ representing this map must be
normalized. As it turns out, this vector has norm larger than one in general. The following lemma says how the map
has to be normalized in order to obtain a Pauli map on $AB$.

\begin{lemma}
\label{LemGlobalPauli}
Let $A$, $B$, and $AB$ be transitive dynamical state spaces, where $A$ and $AB$ are irreducible.
If $X^A:A\to\R$ is a Pauli map on $A$, then the following identity holds and describes a Pauli map on $AB$:
\[
   \frac{X^A\otimes u^B}{\|(X^A\otimes u^B)^\wedge\|_2}=\sqrt{\p(\varphi^A\otimes\mu^B)}\, X^A\otimes u^B,
\]
where $\varphi^A$ is an arbitrary pure state on $A$, and $\mu^B$ is the maximally mixed state on $B$.
\end{lemma}
\proof
Let $c:=\|(X^A\otimes u^B)^\wedge\|_2$ and $X:=\frac 1 c X^A \otimes u^B$, then clearly $\|\hat X\|_2=1$ and
$X(\mu^{AB})=\frac 1 c X^A(\mu^A) u^B(\mu^B)=0$, so $X$ is a Pauli map according to Definition~\ref{DefPauli}.
It remains to show that $\frac 1 {c^2}=\p(\varphi^A\otimes\mu^B)$. Recall the decomposition of $(AB)^\wedge$ from eq.~(\ref{eqDecomp}).
The functional $X^A\otimes u^B$ acts as the zero map on $\hat A\otimes\hat B$ and $\mu^A\otimes \hat B$, hence it achieves its maximal
value on unit vectors on the subspace $\hat A\otimes\mu^B$. Thus, by elementary analysis,
\[
   c=\|(X^A\otimes u^B)^\wedge\|_2=\max_{\hat\varphi\in(AB)^\wedge\setminus\{0\}} \frac{|X^A\otimes u^B(\hat \varphi)|}{\|\hat\varphi\|_2}
   =\max_{\hat\varphi\in \hat A\otimes \mu^B\setminus\{0\}} \frac{|X^A\otimes u^B(\hat \varphi)|}{\|\hat\varphi\|_2}
   =\max_{\hat a\in \hat A\setminus\{0\}}\frac{|X^A\otimes u^B(\hat a \otimes\mu^B)|}{\|\hat a\otimes\mu^B\|_2}.
\]
According to Lemma~\ref{LemInnerProdMaxMix}, we have $\|\hat a\otimes\mu^B\|_2=\sqrt{\langle\hat a\otimes\mu^B,\hat a\otimes\mu^B\rangle}
=\sqrt{\p(\varphi^A\otimes\mu^B)}\|\hat a\|_2$, hence
\[
   c=\frac 1 {\sqrt{\p(\varphi^A\otimes\mu^B)}} \max_{\hat a\in\hat A \setminus\{0\}}\frac{|X^A(\hat a)|}{\|\hat a\|_2}
   =\frac{\|\hat X^A\|_2}{\sqrt{\p(\varphi^A\otimes\mu^B)}}=\frac 1{\sqrt{\p(\varphi^A\otimes\mu^B)}}.
\]
This proves the claim.
\qed

Suppose we draw a pure state $\omega^{AB}\in\Omega_{AB}$ at random. This can be alternatively understood as a two-part process:
first, we fix an arbitrary pure state $\alpha^{AB}\in\Omega_{AB}$. Then, we apply a random reversible transformation $T\in\G_{AB}$
to it (drawn according to the Haar measure): the result $\omega^{AB}=T\alpha^{AB}$ will be a random pure state.
Similarly, we can fix a mixed state
$\alpha^{AB}\in\Omega_{AB}$ with $\p(\alpha^{AB})=\p_0<1$, and apply
a Haar-random reversible transformation to it: $\omega^{AB}=T\alpha^{AB}$.

Having an initially mixed global state describes, for example, classical coin tossing, with $A$ the coin and $B$ the environment.
We will loosely describe this situation as ``drawing a random state $\omega^{AB}$ of purity $\p_0:=\p(\alpha^{AB})=\p(\omega^{AB})$'',
but this description is not quite correct: there is no natural invariant measure on the set
of all states with fixed purity $\p_0<1$, since those states are in general not all connected by reversible transformations (an obvious example
is given by the square state space in Figure~\ref{fig_gbit}).
Thus, not all properties of the random state $\omega^{AB}$ will be independent of the initial state $\alpha^{AB}$. However, as we
shall see, the expected local purity \emph{will} be independent of $\alpha^{AB}$, and this is all we are interested in here.

\begin{theorem}
\label{TheMain1}
Let $A$, $B$, and $AB$ be transitive dynamical state spaces, where $A$ and $AB$ are irreducible. Draw a state
$\omega^{AB}\in\Omega_{AB}$ of fixed purity $\p(\omega^{AB})$ randomly. Then, the expected purity of the local reduced state $\omega^A$ is
\[
   \mathbb{E}_\omega \p(\omega^A)=\frac{K_A-1}{K_A K_B-1}\cdot\frac{\p(\omega^{AB})}{\p(\varphi^A\otimes \mu^B)},
\]
where $\varphi^A$ is an arbitrary pure state on $A$, and $\mu^B$ is the maximally mixed state on $B$.
\end{theorem}
\proof
Let $X^A$ be any Pauli map on $A$, then $X:=\sqrt{\p(\varphi^A\otimes\mu^B)}\, X^A\otimes u^B$ is a Pauli map
on $AB$ according to Lemma~\ref{LemGlobalPauli}. Using the invariance of the Haar measure
and Lemma~\ref{LemFormulaPauliPurity}, we calculate
\begin{eqnarray*}
   \mathbb{E}_\omega \frac{\p(\omega^A)}{K_A-1} &=& \mathbb{E}_\omega \int_{G\in\G_A} \left(X^A\circ G(\omega^A)\right)^2\, dG
   =\int_{G\in\G_A}\mathbb{E}_\omega \left( X^A\otimes u^B(G\otimes\Id(\omega^{AB}))\right)^2\, dG\\
   &=&\mathbb{E}_\omega\left(X^A\otimes u^B(\omega^{AB})\right)^2
   =\int_{G\in\G_{AB}} \mathbb{E}_\omega \left(X^A\otimes u^B(G\omega^{AB})\right)^2\, dG\\
   &=&\frac 1 {\p(\varphi^A\otimes \mu^B)}\, \mathbb{E}_\omega\int_{G\in\G_{AB}} \left(X\circ G(\omega^{AB})\right)^2\, dG
   = \frac 1 {\p(\varphi^A\otimes \mu^B)}\cdot \frac{\mathbb{E}_\omega \p(\omega^{AB})}{K_A K_B-1}.
\end{eqnarray*}
The symbol $\mathbb{E}_\omega$ denoting the expected value with respect to $\omega$ disappears since
$\omega^{AB}$ is drawn from a uniform distribution of a set of states with
fixed purity (as described before the lemma).
\qed

It is easily checked that this result contains the well-known quantum result that random pure bipartite states are
almost maximally entangled with high probability if $|B|\gg |A|$, but we will not demonstrate this here. Instead,
we ask whether the expression $\p(\varphi^A\otimes \mu^B)$ appearing in Theorem~\ref{TheMain1} can be simplified.
It turns out that this is possible under some additional assumptions, and that this expression is related to the
\emph{information carrying capacity} of the involved state spaces. This will be shown in the next subsection.

\subsection{Classical subsystems and capacity}
\label{SubsecCapacity}
How can we quantify the ability of a system (or state space) to carry classical information? A classical bit $A$ and a quantum
bit $B$ both carry one bit of classical information, even though the state space dimensions are quite different: $K_A=2$,
while $K_B=4$. The relevant quantity turns out to be the maximal number of perfectly distinguishable states~\cite{Hardy01}, denoted $N$.
In order to define it, we have to talk about measurements.

Single measurement outcomes on a state space $A$ are described by \emph{effects}, which are linear maps $E:A\to\R$ with the
property that $E(\omega)\geq 0$ for all $\omega\in A_+$. The set of all effects~\footnote{We will not consider the possibility
to have only a \emph{subcone} of $A_+^*$ as the cone of allowed effects. This more general setting would describe
a situation where some mathematically well-defined effects are physically impossible to measure, similar to ``superselection
rules'' forbidding certain superpositions in quantum mechanics. In this paper, we assume that all effects can be physically implemented.
} is known as the \emph{dual cone} $A_+^*$ of the cone of unnormalized states $A_+$ in
convex geometry~\cite{Aliprantis}. An \emph{$n$-outcome measurement} is a collection of effects $E_1,\ldots,E_n$ that sum to the
order unit: $\sum_{i=1}^n E_i=u^A$. The probability of obtaining outcome $i$ on state $\omega\in\Omega_A$ is $E_i(\omega)$.

\begin{definition}
\label{DefClassSubs}
Let $A$ be any state space.
\begin{itemize}
\item A set of pure states $\omega_1,\ldots,\omega_n\in\Omega_A$ is called a \emph{classical subsystem} if there is a measurement
$E_1,\ldots,E_n$ such that $E_i(\omega_j)=\delta_{ij}$ (which is $1$ for $i=j$ and $0$ otherwise); that is, if the states are
perfectly distinguishable by a single-shot measurement.
\item The \emph{capacity} $N_A$ is defined to be the maximal size of any classical subsystem of $A$.
\item If $A$ is a transitive state space, then a classical subsystem $\omega_1,\ldots,\omega_n$ will be called \emph{centered}
if $\displaystyle \frac 1 n \sum_{i=1}^n \omega_i=\mu^A$.
\item  If $A$ is a transitive dynamical state space, then a classical subsystem $\omega_1,\ldots,\omega_n$ will be
called \emph{dynamical} if for every permutation $\pi$ on $\{1,\ldots,n\}$, there is a reversible transformation $T_\pi\in\G_A$
such that $T_\pi(\omega_i)=\omega_{\pi(i)}$ for all $i$.
\end{itemize}
\end{definition}
A ``classical subsystem'' is a subset of a state space which, in many respects, behaves like a classical system from
probability theory. For example, given orthonormal vectors $|\psi_1\rangle,\ldots,|\psi_n\rangle\in\C^d$ with $\langle \psi_i|\psi_j\rangle=\delta_{ij}$,
the corresponding quantum states $\omega_i:=|\psi_i\rangle\langle\psi_i|$ constitute a classical subsystem. It is centered if and only if
$n=d$, and the quantum state space capacity is its Hilbert space dimension $d$. A classical subsystem is \emph{dynamical}
if it also carries all of the reversible dynamics of classical probability theory -- that is, all the permutations. This is clearly the case in
quantum theory, where every permutation of the orthonormal basis vectors can be implemented by a unitary transformation.

Is is easy to see that to any set of mixed states $\omega_1,\ldots,\omega_n$ with effects $E_1,\ldots,E_n$ such that $E_i(\omega_j)=\delta_{ij}$,
there exists a set of \emph{pure} states $\omega'_1,\ldots,\omega'_n$ such that $E_i(\omega'_j)=\delta_{ij}$. Thus, the requirement
of purity in this definition introduces no restriction.
Here are some simple consequences of this definition:
\begin{lemma}
\label{LemPropClassSubs}
We have the following properties of capacity and classical subsystems:
\begin{itemize}
\item[(i)] Capacity satisfies $N_A\leq K_A$, and we have equality if and only if $A$ is a classical state space, i.e.\ $\hat\Omega_A$
is a simplex.
\item[(ii)] If $\omega_1,\ldots,\omega_n$ is a centered classical subsystem, then necessarily $n=N_A$.
\item[(iii)] If $A$ and $B$ carry centered classical subsystems, then so does $AB$, and we have $N_{AB}=N_A N_B$.
\item[(iv)] If a dynamical classical subsystem contains the maximally mixed state in its affine hull, then it is centered.
\end{itemize}
\end{lemma}
\proof
(i) It follows from the definition that sets of perfectly distinguishable states are linearly independent. Since the number of linearly independent
vectors is upper-bounded by the dimension $K_A$, this proves that $N_A\leq K_A$. Now suppose we have equality, then the perfectly
distinguishable states $\omega_1,\ldots,\omega_n$ with $n=N_A$ are a basis of $A$. Every state $\omega\in\Omega_A$ can thus
be written $\omega=\sum_{i=1}^{N_A} \alpha_i \omega_i$ with $\alpha_i\in\R$. Since $u^A(\omega)=1=u^A(\omega_i)$, we get
$\sum_{i=1}^{N_A} \alpha_i=1$. Moreover, we have $0\leq E_j(\omega)=\sum_{i=1}^{N_A}\alpha_i E_j(\omega_i)=\alpha_j$.
That is, $\omega$ is in the convex hull of $\omega_1,\ldots,\omega_n$; in other words, $\Omega_A$ is the simplex generated
by the $\omega_i$.

(ii) Clearly, $N_A\geq n$. In order to see the converse inequality, let $\alpha_1,\ldots,\alpha_{N_A}$ be a maximal classical subsystem with
corresponding effects $E_1,\ldots,E_{N_A}$. Due to transitivity, for every $k$, there is a reversible transformation $T_k\in\G_A$ such
that $T_k\omega_1=\alpha_k$. Using the invariance of the maximally mixed state, we get
\[
   E_k(\mu^A)=E_k(T_k\mu^A)=E_k\left( T_k\frac 1 n \sum_{i=1}^n \omega_i\right)=\frac 1 n \sum_{i=1}^n E_k(T_k\omega_i)
   \geq \frac 1 n E_k(T_k\omega_1)=\frac 1 n E_k(\alpha_k)=\frac 1 n.
\]
On the other hand, we have $\displaystyle 1=\sum_{k=1}^{N_A} E_k(\mu^A)\geq \frac{N_A} n$. This proves that $N_A\leq n$.

(iii) If $\{\omega_i^A\}_{i=1}^{N_A}$ and $\{\omega_j^B\}_{j=1}^{N_B}$ are centered classical subsystems on $A$ and $B$ respectively,
then all states $\omega_i^A\otimes\omega_j^B$ are pure. Moreover, they are perfectly distinguishable by the corresponding
product measurement, and
\[
   \frac 1 {N_A N_B}\sum_{i,j} \omega_i^A\otimes \omega_j^B = \left(\frac 1 {N_A} \sum_{i=1}^{N_A}\omega_i^A\right)
   \otimes\left(\frac 1 {N_B}\sum_{j=1}^{N_B} \omega_j^B\right)=\mu^A\otimes\mu^B=\mu^{AB}.
\]
Thus, $\{\omega_i^A\otimes\omega_j^B\}_{i,j}$ is a centered classical subsystem on $AB$ of size $N_A N_B$, and it follows
from part (ii) that $N_{AB}=N_A N_B$.

(iv) Suppose that $\omega_1,\ldots,\omega_n$ is a dynamical classical subsystem on $A$, and $\displaystyle \mu^A=
\sum_{i=1}^n r_i \omega_i$ for some real numbers $r_i\in\R$ with $\sum_{i=1}^n r_i=1$. Let $E_1,\ldots,E_n$ be the corresponding
effects with $E_i(\omega_j)=\delta_{ij}$. Let $j,k\in\{1,\ldots,n\}$ be arbitrary, and $\pi$ a permutation with $k=\pi^{-1}(j)$. Then
\[
   r_j=\sum_{i=1}^n r_i \delta_{ij} = \sum_{i=1}^n r_i E_j(\omega_i)=E_j(\mu^A)=E_j(T_\pi \mu^A)
   =\sum_{i=1}^n r_i E_j(T_\pi \omega_i) = \sum_{i=1}^n r_i E_j(\omega_{\pi(i)})=r_{\pi^{-1}(j)} = r_k.
\]
Thus, all $r_i$ are equal, and since $\sum_{i=1}^n r_i=1$, we must have $r_i=\frac 1 n$ for all $i$.
This proves the claim.
\qed

Not every state space carries a centered classical subsystem. This is illustrated in Figure~\ref{fig_gbit}: both the square state space
$A$ and the pentagon $B$ have capacity $N_A=N_B=2$. Any pair of antipodal pure states of the square constitutes a centered
classical subsystem of $A$, but the pentagon does not possess any centered classical subsystem. A polygonal state space with
$n\geq 4$ sides carries a centered classical subsystem if and only if $n$ is even.

Why is it natural to assume the existence of a \emph{centered} classical subsystem? We will now discuss three good reasons for
a centered classical subsystem to exist in physically relevant state spaces. A first motivation comes from dynamical considerations in group
theory. Consider a qubit. The north and south pole (say,
$\omega_1=|0\rangle\langle 0|$ and $\omega_2=|1\rangle\langle 1|$) constitute a classical subsystem -- and it is one with
rich dynamics: we can do ``classical computation'' in this subsystem, that is, implement all the permutations
(which is just a bit flip in the case of a qubit, but involves many more transformations for higher-dimensional quantum systems).

More generally, we may ask what transformations preserve this classical subsystem. Together with the bit flips,
these are the rotations around the $z$-axis, and there are many of them: only the maximally mixed state (and no other)
is preserved by those transformations. It turns out that this property forces the classical subsystem to be centered:

\begin{lemma}
Let $\omega_1,\ldots,\omega_n$ be a classical subsystem on a transitive dynamical state space $A$,
and let $\G_\omega\subseteq\G_A$ be its stabilizer subgroup. If the maximally mixed state
$\mu^A$ is the only $\G_\omega$-invariant state, then the subsystem is centered.
\end{lemma}
\proof
Every $T\in\G_\omega$ preserves $\mu_\omega:=\frac 1 n\sum_{i=1}^n \omega_i$. If the lemma's condition is
satisfied, this must be the maximally mixed state $\mu^A$.
\qed

This means that if the state space is symmetric enough to allow for a rich group of dynamics (leaving the classical
subsystem invariant), and if that subsystem is ``large'' enough such that the corresponding group ``mixes'' basically
all of state space, then the subsystem must be centered.

As a second motivation, consider any maximal classical subsystem $\omega_1,\ldots,\omega_N$.
We can think of the convex hull ${\rm conv}\{\omega_1,\ldots,\omega_N\}$ as a classical state space (a simplex)
embedded in the more general, larger state space. This simplex carries its own ``classical'' maximally mixed state,
which is $\mu^{\rm classical}:=\frac 1 N\sum_{i=1}^N \omega_i$. The property of being \emph{centered} just means
that this classical maximally mixed state equals the maximally mixed state of the larger theory, $\mu^{\rm classical}=\mu$:
classical probability theory is embedded in a ``symmetric'' way.

From a physics point of view, this is to expect whenever we have some kind of ``decoherence mechanism'' which effectively
reduces observations to the embedded classical system. On an $n$-level quantum system, for example, decoherence can
effectively reduce the observable state space to that of an $n$-simplex, which corresponds to diagonal density matrices in
the Hamiltonian's eigenbasis. Now suppose that decoherence has taken place, and \emph{in addition}, we have total
ignorance about the \emph{classical} state of our system, such that we hold the state $\mu^{\rm classical}$.

Physically, we expect that we are left with \emph{no remaining information at all}: if we have perfect decoherence, followed by
perfect classical ignorance of the state, there should be no more remaining information that we could read out by measurement.
This implies that $\mu^{\rm classical}=\mu$; that is, the existence of a centered classical subsystem. This subsystem determines
a ``preferred basis'' for decoherence.

A third, more operational way to understand this property is a principle\cite{LluisPrivate} of ``information saturation'': suppose that
Alice obtains a message $i\in\{1,\ldots,N\}$ randomly, with uniform distribution. She encodes this message into the state $\omega_i$
of the state space's maximal classical subsystem, and sends it to Bob. The principle of information saturation asserts that Alice
can use this to send the message $i$ to Bob with perfect success probability, \emph{but not more}. This amounts to saying that the
mixed state that she effectively sends, $\frac 1 N \sum_{i=1}^N \omega_i$, should be the maximally mixed state $\mu$ of the theory.

Before turning to the main result of this section, we need to consider one more property of state spaces. So far, we have mainly talked
about classical subsystems on \emph{single} state spaces. However, if we are interested in classical subsystems on \emph{composite}
state spaces $AB$, we expect that our theory can imitate another computational feature of classical probability theory: that
dynamical classical subsystems on $A$ and $B$ combine to \emph{dynamical} classical subsystems on $AB$. In other words, we expect
that $AB$ carries a dynamical classical subsystem which can be \emph{decomposed} into $A$- and $B$-parts.

\begin{definition}[Composite Classical Subsystem]
A composite transitive dynamical state space $AB$ is said to carry a \emph{composite classical subsystem}
if there are centered dynamical classical subsystems $\omega^A_1,\ldots,\omega^A_{N_A}$ on $A$
and $\omega^B_1,\ldots,\omega^B_{N_B}$ on $B$ such that the corresponding classical subsystem
containing the states $\omega^{AB}_{i,j}:=\omega^A_i \otimes \omega^B_j$ is dynamical.
\end{definition}

We know from Lemma~\ref{LemPropClassSubs} that the states $\omega^{AB}_{i,j}$ are automatically a centered
classical subsystem, and $N_{AB}=N_A N_B$. However, it is not automatically clear that all permutations on this
classical subsystems can be implemented reversibly, that is, that this classical subsystem is \emph{dynamical}.
If it is, it will be called a composite classical subsystem.

Intuitively, this means that $A$ and $B$ contain classical probability distributions as subsystems, in the ``friendliest''
possible way: all permutations can be applied; the classical states of $AB$ are combinations of those
of $A$ and $B$; the local ``classical'' maximally mixed states correspond to the maximally mixed states of $A$ and $B$.
The philosophy of this assumption is that physical state spaces should always be generalizations of classical probability theory,
reducing to the latter in the case of decoherence.

Centered dynamical classical subsystems have a nice symmetry property:
\begin{lemma}
\label{LemSymmClass}
Let $\{\omega_1,\ldots,\omega_N\}$ be a centered dynamical classical subsystem on some state space. Then
$\displaystyle \langle \hat\omega_i,\hat\omega_j\rangle=-\frac 1 {N-1}$ for all $i\neq j$.
\end{lemma}
\proof
By definition, for every permutation $\pi$ on $\{1,\ldots,N\}$, there exists a reversible transformation $T_\pi$ such
that $T_\pi \omega_i =\omega_{\pi(i)}$. Hence $\langle \hat\omega_i,\hat\omega_j\rangle=\langle T_\pi \hat\omega_i,
T_\pi\hat \omega_j\rangle=\langle \hat \omega_{\pi(i)},\hat\omega_{\pi(j)}\rangle$. This proves that there is some constant
$\xi\in\R$ such that $\langle \hat\omega_i,\hat\omega_j\rangle=\xi$ for all $i\neq j$. Now use the fact that the classical subsystem
is centered:
\[
   0=\langle\hat\mu,\hat\mu\rangle=\frac 1 {N^2} \sum_{i,j=1}^N \langle \hat\omega_i,\hat\omega_j\rangle
   =\frac 1 {N^2}\left(
      \sum_{i=1}^N \langle\hat\omega_i,\hat\omega_i\rangle+\sum_{i=1}^N \sum_{j\neq i} \langle \hat\omega_i,\hat\omega_j\rangle
   \right)=\frac 1 {N^2}\left(\strut N+N(N-1)\xi\right).
\]
This equation can be used to infer that $\xi=-1/(N-1)$.
\qed

Now we are ready to prove the main result of this subsection.
\begin{theorem}
\label{TheCompositePurity}
Let $A$, $B$, and $AB$ be irreducible, and suppose that $AB$ carries a composite classical subsystem. Then
\[
   \p(\varphi^A\otimes\mu^B)=\frac{N_A-1}{N_A N_B-1}\qquad\mbox{for every pure state }\varphi^A\in\Omega_A.
\]
\end{theorem}
\proof
By definition, there are centered dynamical classical subsystems $\omega^A_1,\ldots,\omega^A_{N_A}$ on $A$, and
$\omega^B_1,\ldots,\omega^B_{N_B}$ on $B$ such that the states $\omega^{AB}_{i,j}:=\omega^A_i\otimes\omega^B_j$
constitute a centered dynamical classical subsystem on $AB$. We know from Lemma~\ref{LemSymmClass} that
$\langle \hat\omega^{AB}_{i,j},\hat\omega^{AB}_{k,l}\rangle=-1/(N_A N_B-1)$ if $(i,j)\neq (k,l)$. Decomposing $\mu^B$,
we get $\omega^A_1\otimes\mu^B=\frac 1 {N_B} \sum_{j=1}^{N_B} \omega_1^A\otimes\omega_j^B$ and thus
$(\omega^A_1\otimes\mu^B)^\wedge=\frac 1 {N_B} \sum_{j=1}^{N_B} (\omega_1^A\otimes\omega_j^B)^\wedge$. Consequently,
\begin{eqnarray*}
   \p(\omega_1^A\otimes \mu^B)&=&\langle (\omega_1^A\otimes\mu^B)^\wedge,(\omega_1^A\otimes\mu^B)^\wedge\rangle
   =\frac 1 {N_B^2} \sum_{j,k=1}^{N_B} \langle (\omega_1^A\otimes \omega_j^B)^\wedge,(\omega_1^A\otimes\omega_k^B)^\wedge\rangle\\
   &=&\frac 1 {N_B^2}\left(
      \sum_{j=1}^{N_B} \langle \hat\omega_{1j}^{AB},\hat\omega_{1j}^{AB}\rangle + \sum_{j=1}^{N_B} \sum_{k\neq j}
      \langle\hat\omega_{1j}^{AB},\hat\omega_{1k}^{AB}\rangle
   \right)=\frac 1 {N_B^2}\left(N_B\cdot 1 +N_B(N_B-1)\left(-\frac 1 {N_A N_B-1}\right)\right).
\end{eqnarray*}
Some simplification completes the proof.
\qed

Substituting Theorem~\ref{TheCompositePurity} into Theorem~\ref{TheMain1} proves
\begin{theorem}
\label{TheMainMain}
Let $A$, $B$, and $AB$ be irreducible, and suppose that $AB$ carries a composite classical subsystem.
Draw a state $\omega^{AB}\in\Omega_{AB}$ of fixed purity $\p(\omega^{AB})$ randomly, then
\[
   \mathbb{E}_\omega\p(\omega^A)=\frac{K_A-1}{K_A K_B-1}\cdot\frac{N_A N_B-1}{N_A-1}\cdot\p(\omega^{AB}).
\]
\end{theorem}

This is the sought-for specialization of Theorem~\ref{TheMain1}. Both Theorem~\ref{TheMain1} and Theorem~\ref{TheMainMain}
give explicit expressions for the expected local purity of random bipartite states. While Theorem~\ref{TheMain1} is more general (it does
not assume the existence of a composite classical subsystem), it has the disadvantage of containing a term $\p(\varphi^A\otimes\mu^B)$
with no simple operational meaning. The statement of Theorem~\ref{TheMainMain} is operationally simpler, but makes stronger
assumptions on the state spaces.

Note also that a further simplification may be made in the case where $K=N^r$ for some integer $r$, a class of theories discussed in~\cite{Wootters86,Hardy01}. Then $r$ becomes the only parameter that determines the expected purity of a subsystem and $\mathbb{E}_\omega\p(\omega^A)\approx N_B^{r-1}$ (where the approximation is good if $N\gg 1$ for all systems/  subsystems under consideration).

\subsection{$GG'$-invariant faces: entanglement in symmetric subspaces}
\label{SubsecSym}
So far, we have computed the expected amount of entanglement (that is, the purity of the local reduced state) only
for the case that we draw the initial pure state from the full state space $AB$. In many cases, however, it is useful
to consider drawing random states \emph{under constraints}. As a paradigmatic ensemble, suppose we draw
a random pure quantum state $|\psi\rangle$ from the symmetric or antisymmetric subspace of $\C^n\otimes\C^n$. What can we
say about the expected local purity in this case?

We will see that Hilbert subspaces correspond to faces of the state space in the sense of
convex geometry. This will enable us to compute the average reduced purity with geometric methods, using
the invariant inner product introduced in earlier subsections.
Moreover, both symmetric and antisymmetric subspace are invariant under all transformations of the form $U\otimes U$.
This behaviour is a special case of the following general-probabilistic definition.
\begin{definition}[$GG'$-invariant face]
Let $AB$ be a composite dynamical state space. A face $\hat\F$ of $\hat\Omega_{AB}$ will be called \emph{$GG'$-invariant}
if for every $G\in\G_A$ there is some $G'\in\G_B$ such that $G\otimes G'$ maps $\F$ into itself.
The stabilizer subgroup $\{G\in\G_{AB}\,\,|\,\, G\F=\F\}$ will be called $\G_\F$. The face $\F$ will be called \emph{transitive}
if for every pair of extreme points (pure states) $\alpha,\omega\in\F$ there is some $G\in\G_\F$ such that $G\alpha=\omega$.
If $\F$ is transitive, we define the \emph{$\F$-maximally mixed state} $\mu_\F$ as
\[
   \mu_\F:=\int_{G\in\G_\F} G\omega\, dG,
\]
where $\omega$ is any pure state in $\F$.
For every $\omega\in\F$, we set $\bar\omega:=\omega-\mu_\F$, and $\bar\F:=\{\bar\omega\,\,|\,\,\omega\in\F\}$.
$\F$ is called \emph{irreducible} if $\G_\F$ acts irreducibly on $\bar \F$.
\end{definition}
Note that $GG'$-invariance is not a symmetric notion: if for every $G$, there is some $G'$ such that $G\otimes G'$ stabilizes $\F$,
then it is not necessarily the case that to every $G'$, there is some $G$ such that $G\otimes G'$ stabilizes $\F$.

\begin{example}
\label{ExGG}
Here are some examples of transitive irreducible $GG'$-invariant faces:
\begin{itemize}
\item The symmetric subspace $\F_{\rm SYM}$ on $n$-level quantum systems $A$ and $B$. If $\pi$ is the projector onto the
symmetric subspace, then $\F_{\rm SYM}=\{\rho\,\,|\,\, \Tr(\rho\pi)=1\}$. This shows that $\F_{\rm SYM}$ is in fact a face of
the state space on $AB$. If $G=U\cdot U^\dagger\in \G_A$ is some unitary transformation, then $G\otimes G\F_{\rm SYM}=\F_{\rm SYM}$,
so it is $GG'$-invariant with $G'=G$.

There is a
one-to-one correspondence between the symmetric subspace and the Hilbert space $\mathcal{H}:=\C^{n(n+1)/2}$: every state in $\F_{\rm SYM}$
corresponds to a density matrix on $\mathcal{H}$, and every map reversible transformation in $\G_{\F_{\rm SYM}}$ corresponds to
a unitary on $\mathcal{H}$. We know that the unitaries act transitively on $\mathcal{H}$, and we have already shown that this action
is irreducible (cf.\ Lemma~\ref{LemClifford} in the appendix), so $\F_{\rm SYM}$ is transitive and irreducible.
\item The totally antisymmetric subspace in $A\otimes B$, where $A\simeq \C^n$ and $B\simeq \C^n\otimes \C^n$.
If $G=U\cdot U^\dagger$, then this set of quantum states is invariant with respect to $G\otimes G'$, where
$G'=U\otimes U \cdot U^\dagger\otimes U^\dagger$.
\item The face $\F$ of $AB$ with $A=B\simeq\C^n$ which consists only of the maximally entangled state,
$\F=\{|\psi_+\rangle\langle\psi_+|\}$, where $|\psi_+\rangle=\frac 1 {\sqrt{n}} \sum_{i=1}^n |i\rangle\otimes|i\rangle$. It is $U\otimes\bar U$-invariant.
\item \emph{Coin tossing in environment with record.} Suppose we have a classical coin (corresponding to one bit), and
an environment whose state can be described by a bit string of length $n-1$. Initially, the joint system is in an uncorrelated state
$\varphi^{AB}=\varphi^A\otimes \varphi^B$. Since the coin's state is known to use (say, it shows heads), $\varphi^A$ is pure; on the other
hand, we may not have full knowledge about the environment, meaning that $\varphi^B$ is mixed.

In contrast to the usual coin tossing example of Subsection~\ref{SubsecUnify}, we additionally assume that the environment always
contains a perfect record of the coin's state. In other words, if the coin's state is $0$ (or heads), the environment's state must be some
bit string from a set $S_0$; if the coin's state is $1$ (tails), it must be some bit string from a set $S_1$. Both $S_0$ and $S_1$ are
subsets of $\{0,1\}^{n-1}$, have empty intersection, and we assume that they have the same cardinality.

As a consequence, the possible configurations of the joint system are restricted to be either of the form $0s_0$ or $1s_1$, where
$s_0\in S_0$ and $s_1\in S_1$. The possible states (that is, probability distributions) have their full support on those configurations.
This defines a face $\F$ of the joint state space $AB$.

Since permutations can map every configuration of this kind to every other, $\F$ is a transitive. Moreover, it is $GG'$-invariant:
if $G\in\G_A$ is a reversible transformation, there are only two possibilities. First, $G$ is the identity. Then, setting $G'$ also equal
to identity yields a map $G\otimes G'$ which preserves $\F$. Second, $G$ is a bit flip. Then, let $T$ be a permutation which
swaps $S_0$ and $S_1$ (leaving all other strings invariant). Then $G\otimes T$ preserves $\F$.

We will study this scenario further in Example~\ref{ExCCT2} below.
\end{itemize}
\end{example}
The following lemma will be useful.
\begin{lemma}
\label{LemFaceMixture}
If $\F$ is a transitive $GG'$-invariant face, and if $A$ is transitive, then $\mu_\F^A=\mu^A$.
\end{lemma}
\proof
Let $G\in\G_A$ be arbitrary, and let $E^A$ be any effect on $A$, then
\begin{eqnarray*}
   E^A(G^{-1} \mu_\F^A) &=& (E^A\circ G^{-1})\otimes u^B(\mu_\F)=E^A\otimes u^B\left(\strut G^{-1}\otimes\Id(G\otimes G'(\mu_\F))\right)
   =E^A\otimes (u^B\circ G')(\mu_\F)\\
   &=& E^A\otimes u^B(\mu_\F)=E^A(\mu_\F^A).
\end{eqnarray*}
Since this is true for all $E^A$, we must have $G^{-1}\mu_\F^A=\mu_\F^A$. But the only state which is invariant with respect
to all reversible transformations on $A$ is $\mu^A$, hence $\mu_\F^A=\mu^A$.
\qed

Another technical ingredient is this:
\begin{lemma}
\label{LemPerp}
Let $AB$ be a transitive dynamical state space and $\bar F$ a transitive irreducible $GG'$-invariant face. Then
$\bar\F\perp\hat\mu_\F$.
\end{lemma}
\proof
Suppose that $\bar a\in\bar\F$ and $G\in\G_\F$, then
\[
   \langle\bar a,\hat\mu_\F\rangle=\langle\hat G \bar a,\hat G \hat\mu_\F\rangle=\langle\hat G \bar a,\hat\mu_\F\rangle
   =\left\langle \int_{G\in\G_\F}\hat G \bar a\, dG,\hat\mu_\F\right\rangle=\langle 0,\hat\mu_\F\rangle=0.
\]
This proves the claim.
\qed

\begin{theorem}
\label{TheFace}
Let $\F$ be a transitive and irreducible $GG'$-invariant face on an irreducible dynamical state space $AB$, where $A$ is also transitive and irreducible.
Drawing a state $\omega^{AB}\in\F$ with fixed purity $\p(\omega^{AB})$ randomly, the expected local purity is
\[
   \mathbb{E}_\omega^\F \p(\omega^A) = \left\|\pi_{\bar \F} (X^A\otimes u^B)^\wedge\right\|_2^2
   \cdot\frac{K_A-1}{K_\F -1}\cdot\left(\strut\p\left(\omega^{AB}\right)-\p(\mu_\F)\right),
\]
where $K_\F$ denotes the dimension of $\F$, $X^A$ is any Pauli map on $A$, and $\pi_{\bar \F}$ denotes the orthogonal
projection onto the span of $\bar F$ (using the invariant inner product on $(AB)^\wedge$).
\end{theorem}
Taking Lemma~\ref{LemGlobalPauli} into account, it is clear that this theorem reduces to Theorem~\ref{TheMain1}
in the case of $\F=\Omega_{AB}$.
\proof Abbreviate $\omega:=\omega^{AB}$.
Similarly as in Definition~\ref{DefPauli}, call a linear map $X:AB\to\R$ a \emph{Pauli map on $\F$} if $X(\mu_\F)=0$ and
$\langle \bar X,\bar X\rangle=1$, where $\bar X\in\bar\F$ is the vector with $\langle\bar X,\bar\omega\rangle=X(\omega)$
for all $\omega\in\F$. If $X$ is a Pauli map on $\F$, the same calculation as in the proof of Lemma~\ref{LemFormulaPauliPurity}
shows that
\[
   \int_{G\in\G_\F} \left(\strut X\circ G(\omega)\right)^2\, dG=\frac{\langle\bar\omega,\bar\omega\rangle}{K_\F-1}
   \qquad\mbox{for all }\omega\in\F.
\]
Due to Lemma~\ref{LemPerp}, we also have $\langle\hat\omega,\hat\omega\rangle=\langle\bar\omega-\hat\mu_\F,
\bar\omega-\hat\mu_\F\rangle=\langle\bar\omega,\bar\omega\rangle+\langle\hat\mu_\F,\hat\mu_\F\rangle$, hence
$\langle\bar\omega,\bar\omega\rangle=\p(\omega)-\p(\mu_\F)$.
According to Lemma~\ref{LemFaceMixture}, we have $X^A\otimes u^B(\mu_\F)=X^A(\mu_\F^A)=X^A(\mu^A)=0$,
hence $\frac 1 c X^A \otimes u^B$ is a Pauli map on $\F$, where $c=\left\| \overline{X^A\otimes u^B}\right\|_2
=\left\| \pi_{\bar \F}(X^A\otimes u^B)^\wedge\right\|_2$. Similarly as in the proof of Lemma~\ref{TheMain1}, we have for all $\omega\in\F$
\begin{eqnarray*}
   \mathbb{E}_\omega^\F \frac{\p(\omega^A)}{K_A-1}&=&\mathbb{E}_\omega^\F \int_{G\in\G_A} \left(\strut
   X^A\circ G(\omega^A)\right)^2\, dG=\mathbb{E}_\omega^\F \int_{G\in\G_A}\left(\strut X^A\otimes u^B(G\otimes G'(\omega))\right)^2\, dG\\
   &=& \int_{G\in\G_A}\mathbb{E}_\omega^\F\left(\strut X^A\otimes u^B(G\otimes G'(\omega))\right)^2\, dG
   =\mathbb{E}_\omega^\F\left(X^A\otimes u^B(\omega)\right)^2 = \int_{G\in\G_\F}\mathbb{E}_\omega^\F\left(X^A\otimes u^B(\omega)\right)^2\, dG\\
   &=&\mathbb{E}_\omega^\F\, c^2\cdot\int_{G\in\G_\F} \left(\frac 1 c X^A\otimes u^B(G\omega)\right)^2\, dG
   =c^2\frac{\langle\bar\omega,\bar\omega\rangle}{K_\F -1}.
\end{eqnarray*}
Combining all the little results proves the claim.
\qed

In the quantum case, we can give an explicit description of the projector $\pi_{\bar\F}$:
\begin{lemma}
\label{LemProjQ}
Suppose that $A$ is a quantum state space, and $\pi$ is a projector onto some subspace. This subspace
defines a face $\F$ of $\Omega_A$ by $\F=\{\rho\,\,|\,\, \Tr(\pi\rho)=1\}=\{\rho\,\,|\,\,\pi\rho\pi=\rho\}$. Then $\pi_{\bar \F}(M)
=\pi M \pi-\pi \Tr(\pi M \pi)/(\Tr \pi)$.
\end{lemma}
\proof
Define $Q(M):=\pi M \pi-\pi \Tr(\pi M \pi)/(\Tr \pi)$ for all $M\in\hat A$, i.e.\ for all $M$ with $M=M^\dagger$ and $\Tr M=0$.
Clearly, $Q(M)^\dagger=Q(M)$ and $\Tr Q(M)=0$, hence we have a map $Q:\hat A \to \hat A$. Furthermore,
\[
   Q(Q(M))=\pi Q(M)\pi-\frac\pi{\Tr\pi} \Tr(\pi Q(M)\pi)=Q(M)-\frac\pi{\Tr\pi} \Tr(Q(M))=Q(M),
\]
hence $Q$ is a projector. Denote the Hilbert space dimension by $d$, then we get for the inner product on $\hat A$
\[
   \frac{d-1}d\langle M,Q(N)\rangle=\Tr(M Q(N))= \Tr\left[
      M\left(
         \pi N \pi - \frac\pi{\Tr \pi} \Tr(\pi N \pi)
      \right)
   \right]
   =\Tr(M\pi N\pi)-\frac d {d-1}\Tr(M\pi)\Tr(N\pi),
\]
and this expression is symmetric with respect to interchanging $M$ and $N$ (for the first addend due to the cyclicity of the trace).
Thus $Q$ is an orthogonal projector on $\hat A$.

The maximally mixed state $\mu_\F$ on the face is $\mu_\F=\pi/(\Tr \pi)$. Suppose that $M\in\bar\F$, i.e.\ there is some $\rho\in\Omega_A$
such that $M=\bar\rho=\rho-\pi/(\Tr \pi)$. Then direct calculation shows that $Q(M)=M$, i.e.\ $\bar\F\subseteq {\rm ran}\, Q$, and thus
${\rm span}\,\bar\F\subseteq {\rm ran}\, Q$. Now let $m:=\Tr\pi$ (the dimension of the subspace), then the term $\pi M \pi$ in the definition
of $Q$ creates an $m\times m$ block matrix, and the subsequent term $\pi \Tr(\pi M \pi)/(\Tr \pi)$ removes the trace of this block matrix,
leaving $m^2-1$ parameters. Thus, $\dim({\rm ran}\, Q)\leq m^2-1$. On the other hand, density matrices in $\F$ are described by $m^2-1$
parameters, so $\dim({\rm span}\,\bar\F)=m^2-1$. This proves that $\dim({\rm ran}\, Q)\leq \dim({\rm span}\, \bar\F)$. Altogether, this
proves that ${\rm span}\,\bar\F={\rm ran}\, Q$, so that $Q$ is the orthogonal projector onto the span of $\bar\F$ as claimed.
\qed

\begin{theorem}
\label{TheQFace}
Let $S$ be a subspace of dimension $N_S$ on a bipartite quantum state space $AB$ with Hilbert space
dimensions $N_A$ and $N_B$, with the property that for every unitary $U$ on $A$ there is a unitary $U'$ on $B$
such that $U\otimes U' S=S$.
 Drawing a state $\rho^{AB}$ on $S$ with fixed purity $\Tr\left[(\rho^{AB})^2\right]$ randomly, the expected
local quantum purity is
\[
  \mathbb{E}_\rho^S \Tr\left[(\rho^A)^2\right]=\frac 1 {N_A} + \frac{N_A^2-1}{N_S^2-1}\cdot\Tr\left[ \left( \pi(E_A\otimes \Id_B)\pi\right)^2\right] \cdot
  \left(\Tr\left[(\rho^{AB})^2\right]-\frac 1 {N_S}\right),
\]
where $E_A=E_A^\dagger$ is any matrix on $A$ with $\Tr E_A=0$ and $\Tr E_A^2=1$.
\end{theorem}
\proof
The set of states on $AB$ that have full support on $S$ is a face $\F$ on the quantum state space $\Omega_{AB}$.
Since $E_A$ is traceless, we have $\int_U U^\dagger E_A U \, dU=0$. Hence, if $\pi$ is the orthogonal projector onto $S$, we have
\begin{eqnarray*}
  \Tr(\pi(E_A\otimes\Id_B)\pi)&=&\Tr(\pi(E_A\otimes\Id_B)=\Tr\left(\strut U\otimes U' \pi U^\dagger\otimes U'^\dagger (E_A\otimes\Id_B)\right)
  =\Tr\left(\pi U^\dagger\otimes U'^\dagger (E_A\otimes\Id_B)U\otimes U'\right)\\
&=&\Tr\left(\strut\pi(U^\dagger E_A U)\otimes \Id_B\right)=\int_U \Tr\left(\strut\pi(U^\dagger E_A U)\otimes \Id_B\right)\, dU
  =\Tr\left(\strut\pi\left(\int_U U^\dagger E_A U\, dU\right)\otimes \Id_B\right)=0.
\end{eqnarray*}
It is easy to check that $X^A(\rho):=\sqrt{\frac{N_A}{N_A -1}}\Tr(E_A\rho)$ is a Pauli map on $A$.
Thus, $X^A\otimes u^B(\rho)=\sqrt{\frac{N_A}{N_A -1}}\Tr(E_A\otimes\Id_B \rho)$,
and so $\left(X^A\otimes u^B\right)^\wedge=\xi_{AB} E_A\otimes\Id_B$, where $\xi_{AB}=\frac{N_A N_B-1}{N_A N_B}\sqrt{\frac{N_A}{N_A -1}} $.
Using Lemma~\ref{LemProjQ}, this proves that
\[
   \pi_{\bar\F}\left(X^A\otimes u^B\right)^\wedge=\xi_{AB}\, \pi_{\bar F}\left(E_A\otimes\Id_B\right)
   =\xi_{AB}\left(
      \pi(E_A\otimes \Id_B)\pi-\frac{\pi \Tr\left(\strut\pi(E_A\otimes\Id_B)\pi\right)}{\Tr\pi}
   \right),
\]
such that
\[
  \left\|\pi_{\bar\F}\left(X^A\otimes u^B\right)^\wedge\right\|_2^2=\xi_{AB}^2
  \|\pi(E_A\otimes\Id_B)\pi\|_2^2=\xi_{AB}^2 \Tr\left[
     \left(\strut \pi(E_A\otimes\Id_B)\pi\right)^2
  \right]\cdot\frac{N_A N_B}{N_A N_B-1}.
\]
In order to apply Theorem~\ref{TheFace}, note that $K_A=N_A^2$ and $K_\F=N_S^2$, and the maximally mixed state on $\F$ is
$\mu_\F=\pi/(\Tr\pi)$, such that
\[
  \p(\mu_\F)=\frac{N_A N_B}{N_A N_B-1}\Tr(\mu_\F^2)-\frac 1 {N_A N_B-1}=\frac 1 {N_A N_B-1}\left(\frac{N_A N_B}{N_S} -1\right).
\]
Expressing all the purities $\p(\sigma)$ in terms of $\Tr(\sigma^2)$ via eq.~(\ref{eqQuantumPurity})
and some algebraic simplification proves the claim.
\qed

In Theorem~\ref{ThePuritySymm} in Subsection~\ref{SubsecSymm}, we apply this result to compute the average
entanglement in symmetric and antisymmetric quantum subspaces. For the remainder of this subsection, we discuss the case of classical
probability theory.
In this case, we can explicitly compute the norm of the projector appearing in Theorem~\ref{TheFace}:
\begin{lemma}
Suppose that $A$ and $B$ are classical state spaces over $N_A$ and $N_B$ outcomes, and $\F$ is any $GG'$-invariant
face on $AB$, corresponding to $N_\F$ outcomes. Then
$\displaystyle \left\| \pi_{\bar\F}\left(X^A \otimes u^B\right)^\wedge\right\|_2^2 = \frac{N_\F (N_A N_B -1)}{N_A N_B (N_A-1)}$
and $\displaystyle \p(\mu_\F)=\frac{N_A N_B/N_\F -1}{N_A N_B-1}$.
\end{lemma}
\proof
We use Theorem~\ref{TheFace}. First, the maximally mixed state $\mu_\F$ is just the uniform distribution on the classical outcomes
that generate $\F$, that is, a probability vector with $N_\F$ entries equal to $1/N_\F$ and all others zero. Recalling the formula for
purity in the classical case, eq.~(\ref{eqPurityClassical}), gives $\displaystyle \p(\mu_\F)=\frac{N_A N_B/N_\F -1}{N_A N_B-1}$.
Now we apply Theorem~\ref{TheFace} to the special case where the initial state is pure: $\p(\omega^{AB})=1$. Since there are
no entangled states in classical probability theory, we know that $\omega^A$ must be pure as well, i.e.\ $\p(\omega^A)=1$, and so is
its expectation value.
Using that $K=N$ classically, substituting all these identities into the statement of Theorem~\ref{TheFace} yields the norm of the projector.
\qed

Substituting this result back into Theorem~\ref{ThePuritySymm}, we get a very simple statement regarding $GG'$-invariant faces in
classical probability theory. The proof involves only simple algebra and is thus omitted.
\begin{theorem}
\label{TheClassGG}
Suppose that $A$ and $B$ are classical state spaces, and $\F$ is any $GG'$-invariant face on $AB$. If we draw a random state $\omega^{AB}$
in $\F$ of fixed purity $\p(\omega^{AB})$, then the expected purity of the local marginal is
\[
   \mathbb{E}_\omega^\F \p(\omega^A)=\p(\omega^{AB}\upharpoonright_\F),
\]
where the right-hand side denotes the purity of $\omega^{AB}$, computed by treating $\omega^{AB}$ as a state on the smaller
state space $\F$.
\end{theorem}
Explicitly, if $\{\omega^{AB}_j\}_{j=1}^{N_\F}$ denote the entries of the probability vector, then
\[
   \p(\omega^{AB}\upharpoonright_\F) = \frac{N_\F}{N_\F -1} \sum_{j=1}^{N_F} \left(\omega^{AB}_j\right)^2 - \frac 1 {N_\F -1}
\]
(compare this with eq.~(\ref{eqPurityClassical})). The result of Theorem~\ref{TheClassGG}  is no surprise at all: we get the same result in
the unconstrained case, Theorem~\ref{TheMainMain}, where the prefactors are cancelled due to $N=K$.

\begin{example}[Coin tossing in environment with record, part 2]
\label{ExCCT2}
Recall the scenario from the last paragraph of Example~\ref{ExGG}. Does the record in the environment affect the randomization of
the coin? Suppose the coin is initially in the pure state $0$ (or heads). Then the environment's initial state $\varphi^B$ must have
full support on $S_0$; for simplicity, we assume that it is otherwise completely unknown, i.e.\ the uniform mixture over $S_0$.
Applying Theorem~\ref{TheClassGG}, a little calculation shows that
\[
   \mathbb{E}_\omega^\F \p(\omega^A)=\frac 1 {2\# S_0 -1}.
\]
This is exactly the same result as Theorem~\ref{TheMainMain} gives us for an unconstrained environment $B$ with $N_B=\# S_0$. This
is an environment which has half as many possible states as in the first scenario, where the possible environment configurations are
in $S_0\cup S_1$ with cardinality $2\# S_0$. Intuitively, the informed environment loses one bit of randomization power due to redundancy.
The same conclusion holds for correlated initial states.
\end{example}

\subsection{Theories which are not locally tomographic}
\label{SubsecNotLocallyTomo}
In the previous sections, we have considered certain types of composite state spaces: transitive locally tomographic compositions $AB$
of state spaces $A$ and $B$. At present date, there are no known examples of such theories beyond quantum theory and subspaces within it such as classical probability theory. The search for such theories has just started recently, but preliminary results suggest that theories of this kind might be rare~\cite{Singularity}.

On the other hand, it is known that there is a multitude of transitive composite state spaces $AB$ if the requirement of local tomography is
dropped~\cite{HowardCozmin}. As it turns out, some of our results are easily generalized to theories without local tomography. We will sketch
this in this subsection, but leave a more detailed analysis of such theories to future work.
We start with a trial definition of arbitrary compositions of state spaces which need not be locally tomographic.
\begin{definition}
\label{DefComposition}
If $A$ and $B$ are state spaces, a \emph{composition} $AB$ is any state space which can be decomposed as $AB=(A\otimes B)\oplus C$ such
that the following properties hold:
\begin{itemize}
\item If $\omega^A\in\Omega_A$ and $\omega^B\in\Omega_B$, then $\omega^A\otimes\omega^B\in\Omega_{AB}$.
\item For $\omega^{AB}\in\Omega_{AB}$, define the vector $\omega^A$ via $L(\omega^A):= L\otimes u^B (\omega^{AB})$ for all
linear maps $L:A\to\R$. (An analogous definition yields $\omega^B$.) Then $\omega^A\in\Omega_A$ and $\omega^B\in\Omega_B$.
\end{itemize}
Moreover, if $A$ and $B$ are dynamical state spaces, a \emph{dynamical} composition $AB$ is assumed to have the following property:
if $T_A\in\G_A$ and $T_B\in\G_B$, then $(T_A\otimes T_B)\oplus\Id_C\in\G_{AB}$.
\end{definition}
Physically, this means that the global state space $AB$ has some degrees of freedom (collected in $C$) that cannot be accessed locally at
$A$ or $B$, not even by comparing correlations of measurement outcomes.
It follows from the second property that $u^{AB}=u^A\otimes u^B$, because $u^A\otimes u^B$ is a linear functional which gives
unity on all global states.
In this notation, the tensor product of two linear functionals on $A$ and $B$ is assumed to act as the zero functional on $C$, i.e.\
$L^A\otimes L^B\equiv L^A\otimes L^B\oplus 0^C$.

The most famous example of a composite state space which is not locally tomographic is quantum theory over the reals~\cite{HardyWootters2010}:
\begin{example}[Real quantum theory]
Let $A=\{\rho\in\R^{m\times m}\,\,|\,\, \Tr\rho=1, \rho=\rho^T, \rho\geq 0\}$, that is, the set of $(m\times m)$-density matrices with all
real entries. The order unit is
$u^A(\rho)=\Tr\rho$. This is a state space of dimension $K_A=m(m+1)/2$. Similarly, let $B$ be the state space of $(n\times n)$-density matrices
with all real entries. We assume $m,n\geq 2$.

Then, a composition of $A$ and $B$ is given by the set of all $(mn)\times(mn)$-density matrices with all real entries.
Since $K_{AB}>K_A K_B$, this is not a locally tomographic composition, but it is easy to check that it satisfies all the properties of
Definition~\ref{DefComposition}.
\end{example}

Since $A$, $B$, and $AB$ are state spaces in the usual sense, the results of Subsections~\ref{SubsecGenBloch} to~\ref{SubsecPaulis} apply
without any modification. As usual, if $AB$ is transitive, it has a decomposition $AB=(AB)^\wedge\oplus \R\cdot\mu^{AB}$. Moreover,
we claim that the locally inaccessible subspace $C$ is part of the Bloch subspace, $C\subseteq (AB)^\wedge$. To see this, let $c\in C$, then
$u^{AB}(c)=u^A\otimes u^B(c)=0$. In more detail, we have the decomposition
\[
   (AB)^\wedge=(\hat A\otimes\hat B)\oplus (\hat A \otimes\mu^B)\oplus  (\mu^A\otimes\hat B)\oplus C
\]
which follows from the fact that the right-hand side is a subspace $V\subset AB$ of dimension $\dim V=\dim(AB)-1$, and $u^{AB}=u^A\otimes u^B$
evaluates to zero on all vectors of $V$.
Now suppose that $AB$ is irreducible -- then all the addends above are mutually orthogonal in the invariant
inner product on $(AB)^\wedge$. For example, to see that $\hat A\otimes\hat B\perp C$, let $\hat a\in \hat A$, $\hat b\in\hat B$ and $c\in C$, and
compute
\[
   \langle \hat a\otimes\hat b,c\rangle=\langle(T_A\otimes \Id_B\oplus\Id_C)\hat a\otimes\hat b,(T_A\otimes\Id_B\oplus\Id_C)c\rangle
   =\langle T_A\hat a\otimes\hat b,c\rangle=\langle\int_{\G_A} T_A \hat a\, dT_A \otimes \hat b,c\rangle=\langle 0,c\rangle=0,
\]
using the same argumentation as in Subsection~\ref{SubsecBipartite}. How is the maximally mixed state $\mu^{AB}$ on $AB$ related
to $\mu^A$ and $\mu^B$? To answer this question, extend the inner product on $(AB)^\wedge$ to an inner product on all of $AB$:
for $v,w\in AB$ with decomposition $v=\hat v+v_0 \mu^{AB}$ and $w=\hat w+w_0\mu^{AB}$, where $v_0=u^{AB}(v)$ and
$w_0=u^{AB}(w)$, we define
\[
   \langle v,w\rangle:=\langle \hat v,\hat w\rangle+v_0 w_0.
\]
This inner product is clearly invariant with respect to all reversible transformations from $\G_{AB}$, and it is constructed such that
$\mu^{AB}\perp (AB)^\wedge$. Taking into account the orthogonality of subspaces mentioned above, this proves that
\[
   C\oplus\R\cdot\mu^{AB}=\left[ (\hat A\otimes\hat B)\oplus (\hat A\otimes \mu^B)\oplus (\mu^A\otimes \hat B)\right]^\perp.
\]
By integration as above, it is also easy to see that $\mu^A\otimes\mu^B$ is perpendicular to all the three subspaces
$\hat A\otimes\hat B$, $\hat A\otimes \mu^B$, and $\mu^A\otimes\hat B$. Thus, $\mu^A\otimes\mu^B\in C\oplus\R\cdot \mu^{AB}$.
In other words, there is some constant $\xi\in\R$ and vector $\mu^C\in C$ such that $\mu^A\otimes\mu^B=\xi\cdot\mu^{AB}-\mu^C$.
Applying $u^{AB}$ to this equation, using that $u^{AB}(\mu^{AB})=u^{AB}(\mu^A\otimes\mu^B)=1$ and $u^{AB}(\mu^C)=0$, we get
$\xi=1$, and thus
\[
   \mu^{AB}=\mu^A\otimes\mu^B+\mu^C.
\]
The Bloch vector $\mu^C$ can be interpreted as the collection of all locally inaccessible degrees of freedom of the maximally
mixed state on $AB$.
For symmetry reasons, we think it is plausible that $\mu^C=0$ for many theories, but we were unable to prove this in generality.

Following the argumentation in Subsection~\ref{SubsecBipartite}, it is interesting to see that both Lemma~\ref{LemInnerProdMaxMix}
and Lemma~\ref{LemGlobalPauli} remain valid if $AB$ is not locally tomographic with only minor modifications. We now assume that $A$,
$B$, and $AB$ are transitive dynamical state spaces, where $A$ and $AB$ are irreducible. Lemma~\ref{LemInnerProdMaxMix} becomes
\[
   \langle\hat x\otimes\mu^B,\hat y\otimes\mu^B\rangle=\left(\p(\varphi^A\otimes \mu^B)-\|\mu^C\|_2^2\right)\langle \hat x,\hat y\rangle,
\]
while Lemma~\ref{LemGlobalPauli} gets modified to stating that
\[
   \frac{X^A\otimes u^B}{\|(X^A\otimes u^B)^\wedge\|_2}=\sqrt{\p(\varphi^A\otimes\mu^B)-\|\mu^C\|_2^2} \,X^A\otimes u^B
\]
is a Pauli map on $AB$.

 The only modification of Theorem~\ref{TheMain1}
is that $K_A K_B$ has to be replaced by $K_{AB}$, the dimension of the composite state space. The rest of the proof remains unaltered.
\begin{theorem}
Let $A$, $B$, and $AB$ be transitive dynamical state spaces, where $A$ and $AB$ are irreducible, and $AB$ is not necessarily
locally tomographic. Draw a state
$\omega^{AB}\in\Omega_{AB}$ of fixed purity $\p(\omega^{AB})$ randomly. Then, the expected purity of the local reduced state $\omega^A$ is
\[
   \mathbb{E}_\omega \p(\omega^A)=\frac{K_A-1}{K_{AB}-1}\cdot\frac{\p(\omega^{AB})}
   {\p(\varphi^A\otimes \mu^B)-\|\mu^C\|_2^2},
\]
where $\varphi^A$ is an arbitrary pure state on $A$, $\mu^B$ is the maximally mixed state on $B$, and $\mu^C$ is the vector which contains
the locally inaccessible degrees of freedom of the maximally mixed state $\mu^{AB}$ on $AB$.
\end{theorem}
We leave it open whether the results of Subsection~\ref{SubsecCapacity} (including a more operational formulation of the main result
as in Theorem~\ref{TheMainMain} involving only $N$ and $K$) can be generalized to composite state spaces that are not locally tomographic: 
this seems to depend strongly on the question under what circumstances $\mu^{AB}=\mu^A\otimes\mu^B$ remains true such that $\mu^C=0$.

\section{Summary and outlook}
In summary, we considered general probabilistic theories and asked how mixed (impure) subsystems tend to be in such theories after undergoing reversible dynamics. We showed that under certain limited assumptions subsystems tend be close to maximally mixed in appropriate limits, and the amount of purity is given by a simple formula. Showing this involved developing various generalizations of the corresponding quantum concepts, e.g.\ purity, which are of interest in themselves. Our results also apply to subspaces within quantum theory, and we calculated for example the expected purity of subsystems in symmetric and antisymmetric spaces.  
 
We view our results as a significant first step towards formulating the second law as a `meta-theorem', meaning a theorem that requires weaker assumptions than for example all of quantum theory. More generally we envisage a formulation of statistical mechanics independent of theory details. Such a formulation can be expected to be useful for example for black hole thermodynamics, where one cannot be certain that standard quantum theory applies, but may accept some more basic assumptions.

\appendix
\section{Irreducibility of the Clifford group}
\label{AppIrr}
Here we prove a lemma which is used in the main text in Example~\ref{ExQuantumPauli}. It shows that our generalized
definition of a Pauli map reduces to the usual Pauli operators for the case of several qubits in quantum theory. It exploits
the well-known fact that the Clifford group is a 2-design~\cite{DiVincenzo02}.
\begin{lemma} 
 \label{LemClifford}
 The Clifford group $C_k$ on $k$ qubits acts irreducibly by conjugation on the real vector space of traceless Hermitian
 $2^k\times 2^k$-matrices.
 \label{LemCliffordIrrep}
 \end{lemma}
\proof
We use the notation from Example~\ref{ExQuantumPauli}. If there is a real subspace $S\subseteq \hat A$ which is invariant with
respect to all Clifford maps, i.e.\ $U S U^\dagger\subseteq S$ for all $U\in C_k$, then its complexification
\[
   S':=S+i S :=\{S_1+i S_2\,\,|\,\, S_1,S_2\in S\}\subseteq \mathcal{B(H)}
\]
is a complex subspace of the set of all complex matrices $\mathcal{B(H)}$ on the Hilbert space $\mathcal{H}=\left(\C^2\right)^{\otimes k}$
which is also invariant with respect to all Clifford maps. Fix any orthonormal basis $\{|i\rangle\}_{i=1}^{2^k}$ on $\mathcal{H}$, and
define a complex-linear map $\Phi:\mathcal{B}\to \mathcal{H}\otimes\mathcal{H}$ (which is related to the infamous Choi-Jamiolkowski isomorphism) by
\[
   \Phi(M):=\sum_{i,j=1}^{2^k} \langle i|M|j\rangle\,|i\rangle\otimes |j\rangle.
\]
It is a linear isomorphism which satisfies $\langle\Phi(M),\Phi(N)\rangle=\Tr(M^\dagger N)$. Moreover, we have
$\Phi(U M U^\dagger)=(U\otimes \bar U)\Phi(M)$ for all unitaries $U$, where $\bar U$ denotes the complex-conjugate of $U$
with respect to the given basis. It follows that the complex subspaces of $\mathcal{B(H)}$ which are invariant under conjugation
with respect to all unitaries $U\in C_k$
are in one-to-one correspondence with the complex subspaces of $\mathcal{H}\otimes\mathcal{H}$ which are invariant under
$U\otimes\bar U$ for all unitaries $U\in C_k$.

An obvious invariant subspace in $\mathcal{B(H)}$ consists of all complex
multiples of the identity $\id$. The image is $\Phi(\id)=\sqrt{2^k}|\psi_m\rangle=\sum_i |i\rangle\otimes|i\rangle$, which
reproduces the well-known fact that multiples of the maximally entangled state $|\psi_m\rangle$ are invariant with respect
to transformations of the form $U\otimes \bar U$. In order to prove the lemma, we have to show that this subspace and its
orthogonal complement (consisting of the traceless matrices in $\mathcal{B(H)}$ respectively of the vectors that are
orthogonal to $|\psi_m\rangle$) are the only non-trivial subspaces which are $C_k$-invariant.

It is well-known~\cite{DiVincenzo02} that the Clifford group is a $2$-design, i.e.
\begin{equation}
   \frac 1 {|C_k|} \sum_{U \in C_k} (U\otimes U)M(U^\dagger\otimes U^\dagger)
   =\int_{U\in U(2^k)} (U\otimes U)M(U^\dagger\otimes U^\dagger)\, dU
   =\frac{2\Tr(\pi_s M \pi_s)}{2^k(2^k+1)} \pi_s + \frac{2\Tr(\pi_a M \pi_a)}{2^k(2^k -1)}\pi_a
   \label{eq2Design}
\end{equation}
for all $M\in\mathcal{B(H)}$, where $\pi_s$ and $\pi_a$ denote the projectors onto the symmetric and
antisymmetric subspaces of $\mathcal{H}\otimes\mathcal{H}$, respectively. We can write $\pi_s=(\id+\mathbb{F})/2$
and $\pi_a=(\id-\mathbb{F})/2$, where $\mathbb{F}|i\rangle\otimes|j\rangle=|j\rangle\otimes|i\rangle$ is the swap
operator. It is easy to see that $2^k\langle\psi_m|M^{T_B}|\psi_m\rangle=\Tr(M\mathbb{F})$, and
$\mathbb{F}^{T_B}=2^k |\psi_m\rangle\langle\psi_m|$, if $T_B$ denotes the partial transposition on the second system.
Moreover, it holds $(A\otimes B\rho C\otimes T)^{T_B}=A\otimes D^T \rho^{T_B} C\otimes B^T$~\cite{Bruss}.
Using these identities and applying $T_B$ to eq.~(\ref{eq2Design}), we get
\[
   \frac 1 {|C_k|} \sum_{U \in C_k} (U\otimes \bar U)N(U^\dagger\otimes \bar U^\dagger)
   =\int_{U\in U(2^k)} (U\otimes \bar U)N(U^\dagger\otimes \bar U^\dagger)\, dU
   =\frac{\Tr(\pi_m^\perp N \pi_m^\perp)}{2^{2k}-1}\pi_m^\perp +\Tr(\pi_m N \pi_m) \pi_m,
\]
where $\pi_m:=|\psi_m\rangle\langle\psi_m|$ and $\pi_m^\perp:=\id-\pi_m$. By Schur's Lemma, it follows that the
one-dimensional subspace spanned by $|\psi_m\rangle$ and its orthogonal complement are the only non-trivial
subspaces which are invariant with respect to $U\otimes\bar U$ for all $U\in C_k$.
\qed

\section{Purity in boxworld}
\label{SecPurityBoxworld}
As we show here, it is possible to define a notion of purity for \emph{generalized no-signalling theory}~\cite{Barrett07}, colloquially
called \emph{boxworld}, even though this theory is not transitive~\cite{boxworld}. However, it will turn out that the resulting notion
of purity does not have all the nice properties that hold in the transitive case.

To keep things simple, we will only consider the paradigmatic case of two observers (Alice and Bob), each
carrying a square state space (a so-called ``gbit'', as shown in, and discussed around, Figure~\ref{fig_gbit}). Operationally,
this means that both Alice and Bob carry two measurement devices with two outcomes each (``yes'' and ``no'');
local states are characterized by the two probabilities of the ``yes''-outcomes. Both probabilities can be chosen independently,
giving rise to two coordinates in a square state space.

The two local state spaces are equal: $A=B$ (we use the two different labels for convenience). Now we use a particular
representation of the square state space introduced in~\cite{boxworld}.
We define the set of normalized states $\Omega_A$ as the convex hull of the four pure states
\[
   \omega_{\pm\pm}:=\left(\begin{array}{c} 1 \\ \pm 1/\sqrt{2} \\ \pm 1/\sqrt{2} \end{array}\right).
\]
As usual, the cone of unnormalized states is $A_+:=\R_0^+\cdot \Omega_A$, and the order unit is $u^A=(1,0,0)^T$, if we denote
effects by vectors (such that $u^A(\omega)=\langle u^A,\omega\rangle$ in the usual inner product).
It turns out that the cone of effects $A_+^*$ is generated by the four effects
\[
   Y=\left(\begin{array}{c} 1/2 \\ 1/\sqrt{2} \\ 0 \end{array}\right),\qquad u^A-Y=\left(\begin{array}{c} 1/2 \\ -1/\sqrt{2} \\ 0 \end{array}\right),
   \qquad
   Z=\left(\begin{array}{c} 1/2 \\ 0 \\1/\sqrt{2} \end{array}\right),\qquad u^A-Z=\left(\begin{array}{c} 1/2 \\ 0 \\ -1/\sqrt{2} \end{array}\right).
\]
The square state space is transitive. As discussed in Example~\ref{ExPurityGbit}, the group of reversible transformations
$\G_A$ is the dihedral group $D_4$. In the particular representation chosen here, it acts as on the $y$- and $z$-components of
a state vector $\omega$, and leaves the $x$-component (the normalization) invariant. The maximally mixed state is
$\mu^A=\left(\begin{array}{c} 1 \\ 0 \\ 0 \end{array}\right)$. The set of unnormalized states on $AB$ is defined as follows:
\[
   (AB)_+:=\left\{ \omega\in A\otimes B \,\,|\,\, E^A\otimes F^B(\omega)\geq 0\mbox{ for all }E^A\in A_+^*, F^B\in B_+^*\right\}.
\]
That is, these are all the vectors with the property that all local measurements yield positive outcome probabilities. The bipartite
(normalized) state space $\Omega_{AB}$ consists of all $\omega\in (AB)_+$ with $u^{AB}(\omega)=1$, where the order unit
is, as always, $u^{AB}=u^A\otimes u^B$. Since $AB$ is $9$-dimensional, $\Omega_{AB}$ is an $8$-dimensional polytope, known
as the \emph{no-signalling polytope}.

What are the pure states in $AB$? Clearly, the $16$ product states $\omega_{\pm\pm}\otimes \omega_{\pm\pm}$ are pure.
But there are $8$ additional entangled pure states: one of them is the famous PR box state $\omega_{PR}$, and the others
can be obtained by local transformations from $\omega_{PR}$. We could use vectors with $9$ entries to write down those
states explicitly, but it will be more convenient to use another representation: given the three unit vectors $e_1,e_2,e_3$,
we will denote states (and vectors) $\omega\in AB$ as matrices $(\omega_{ij})$, by using the decomposition
$\omega=\sum_{i,j=1}^3 \omega_{ij} e_i \otimes e_j$. In this representation, the maximally mixed state and the pure product
states are
\[
   \mu^{AB}=\left(
      \begin{array}{ccc}
         1 & 0 & 0 \\
         0 & 0 & 0 \\
         0 & 0 & 0
      \end{array}
   \right),\qquad
   \omega_{rs}\otimes \omega_{tu}=\left(
      \begin{array}{ccc}
          1 & t/\sqrt{2} & u/\sqrt{2} \\
          r/\sqrt{2} & rt /2 & ru/2\\
          s/\sqrt{2} & st/2 & su/2
      \end{array}
   \right)\quad (r,s,t,u\in\{-1,+1\}).
\]
One of the PR-box states is
\[
   \omega_{PR}=\left(
      \begin{array}{ccc}
         1 & 0 & 0 \\
         0 & 1/2 & 1/2 \\
         0 & 1/2 & -1/2
      \end{array}
   \right).
\]
What are the reversible transformations in $AB$? It can be shown~\cite{boxworld} that these are exactly the \emph{local}
transformations, that is, those of the form $G_A\otimes G_B$, together with the swap transformation $S$ which exchanges
the two subsystems. There are no other reversible transformations in $\G_{AB}$. The bipartite space $AB$ decomposes
into $\G_{AB}$-invariant subspaces (the first addend cannot be decomposed further because $D_4$ acts complex-irreducibly on $\hat A$):
\begin{equation}
   AB = \underbrace{(\hat A\otimes \hat B)}_{4-\dim.} \oplus \underbrace{(\mu^A\otimes \hat B \oplus \hat A\otimes \mu^B)}_{4-\dim.}
   \oplus \underbrace{(\R\cdot \mu^A\otimes\mu^B)}_{1-\dim.}.
   \label{eqDecompBoxworld}
\end{equation}
In this notation, $\hat A$ denotes the subspace of vectors $x\in A$ with $u^A(x)=0$ (we called this the ``Bloch subspace'' in
Subsection~\ref{SubsecGenBloch}). This shows that the only state on $AB$ which is invariant with
respect to all reversible transformations is the maximally mixed state $\mu^{AB}:=\mu^A\otimes\mu^B$. We call the subspace
generated by the first two addends above $(AB)^\wedge$, such that
\[
   AB=(AB)^\wedge \oplus\R\cdot\mu^{AB}.
\]
In other words, $(AB)^\wedge$ consists of all vectors $x\in AB$ with $u^{AB}(x)=0$.

Now we proceed as in Section~\ref{SecMath}: to every state $\omega\in\Omega_{AB}$, we define the corresponding Bloch vector
$\hat\omega$ as $\hat\omega:=\omega-\mu^{AB}$. This vector is obtained from the matrix representation above by replacing
the ``$1$'' in the upper-left corner by a zero. Denote the usual Euclidean inner
product by $\langle \cdot,\cdot\rangle$. Then we define the purity of $\omega$ as
\begin{equation}
   \p(\omega):=c\cdot \langle\hat\omega,\hat\omega\rangle,
   \label{eqPurityBoxworld}
\end{equation}
and we choose the constant $c>0$ such that, say, the pure product states have purity $\p(\omega_{\pm\pm})=1$. Using the representation
above, it is easy to see that we must have $c=1/3$. The resulting definition satisfies some of the properties mentioned in
Lemma~\ref{LemPropPur}:
\begin{itemize}
\item $0\leq \p(\omega)\leq 1$ for all $\omega\in\Omega_{AB}$,
\item $\p(\omega)=0$ if and only if $\omega=\mu^{AB}$, i.e.\ if $\omega$ is the maximally mixed state on $AB$,
\item $\sqrt{\p}$ is convex, and
\item $\p(T\omega)=\p(\omega)$ for all reversible transformations $T\in\G_{AB}$ and states $\omega\in\Omega_{AB}$.
\end{itemize}
For example, to prove the last point, note that the local transformations on $A$ and $B$ are orthogonal in the chosen
representation: they rotate and reflect the square. Hence their product is orthogonal as well, and so is the swap.
It follows that $\p(T\omega)=\langle T\hat\omega,T\hat\omega\rangle=\langle\hat\omega,T^\dagger T \hat\omega\rangle
=\langle\hat\omega,\hat\omega\rangle=\p(\omega)$.

However, there is some bad news: if we compute the purity of the pure PR box state, we get
\[
   \p(\omega_{PR})=\frac 1 3 \langle \hat\omega_{PR},\hat\omega_{PR}\rangle=\frac 1 3.
\]
Even though this state is pure, it has purity (much) less than one. On transitive state spaces as considered in Section~\ref{SecMath},
this cannot happen: all pure states have purity $1$. Vice versa, we can see from this result that there is no reversible
transformation which maps a pure product state to a PR-box state: if there was one, then both states necessarily would have the same purity.

Can we somehow avoid this problem? So far, we have been a bit hasty in our definition: in eq.~(\ref{eqPurityBoxworld}), we defined
purity with respect to the usual Euclidean inner product, because all reversible transformations in $\G_{AB}$ are orthogonal with respect to
this inner product. However, the decomposition~(\ref{eqDecompBoxworld}) shows that this is not the only inner product on $(AB)^\wedge$
(where the Bloch vectors $\hat\omega$ live) which has this property: if we have two vectors $\hat\varphi,\hat\omega$ in this space,
we can decompose them as
\[
   \hat\varphi=\varphi'+\varphi'',\qquad \varphi'\in \hat A\otimes\hat B,\qquad \varphi''\in (\mu^A\otimes\hat B)\oplus(\hat A\otimes\mu^B)
\]
and similarly for $\hat\omega$, and then define
\[
   \langle \hat\varphi,\hat\omega\rangle:=a\langle \varphi',\omega'\rangle+b\langle\varphi'',\omega''\rangle,
\]
where all brackets on the right-hand side denote the usual Euclidean inner product.
For every choice of $a,b>0$, this yields an invariant inner product on $(AB)^\wedge$. Is there a way to choose $a$ and $b$ such that
the purity of pure product states and PR-box states equals unity at the same time? (We can retain $c=1/3$ in definition~(\ref{eqPurityBoxworld})
and absorb any necessary factor into $a$ and $b$). Using that $\mu^A=e_1$ and $\hat A ={\rm span}\{e_2,e_3\}$, it is easy to see that
\[
   \p(\omega_{\pm\pm})=\frac 1 3(a+2b),\qquad
   \p(\omega_{PR})=\frac 1 3 \cdot a.
\]
Both expressions can only be simultaneously equal to $1$ if $b=0$. But this ruins the inner-product property. If we ignore this,
and go on with setting $a=3$ and $b=0$, we loose the property that $\p(\omega)=0$ only for the maximally mixed state $\omega=\mu^{AB}$:
for example, we get $\p(\mu^A\otimes\omega_\pm)=0$.

In summary, there is no definition of purity on bipartite boxworld which has all the nice properties that hold true in transitive state spaces.
However, if we accept the existence of pure states with purity less than one, then eq.~(\ref{eqPurityBoxworld}) can be a useful definition.
A similar conclusion holds for other non-transitive state spaces.

{\em \bf Acknowledgements.---} We acknowledge valuable discussions with Fr\'ed\'eric Dupuis, Johan {\AA}berg, Jonathan Oppenheim, L\'idia del Rio, Lucien Hardy, Renato Renner, Roger Colbeck,
as well as financial support from the National Research Foundation (Singapore), the Ministry of Education (Singapore), Swiss National Science Foundation (grant no.\ 200021-119868).
Research at Perimeter Institute is supported by the Government of Canada through Industry Canada and by the Province of Ontario through the Ministry of Research and Innovation.

\bibliography{typentrefs}

\end{document}